\documentclass[aps,pra,twocolumn,nofootinbib,superscriptaddress,10pt,floatfix,longbibliography]{revtex4-2}

\usepackage{float}
\usepackage{minitoc}
\usepackage[toc,page,header]{appendix}
\usepackage{physics}
\usepackage{xfrac}
\usepackage{graphicx}
\usepackage{subfigure}
\usepackage{mathdots}
\usepackage{amsfonts,amssymb,amsmath}
\usepackage{bm}
\usepackage{mathrsfs}
\usepackage[]{graphics,graphicx,epsfig}
\usepackage{amsthm}
\usepackage{csquotes}
\MakeOuterQuote{"}
\usepackage{epstopdf}
\usepackage{tikz}
\usepackage[ruled,linesnumbered]{algorithm2e}
\usepackage{paralist}
\usepackage{diagbox}
\usepackage[inline]{enumitem}
\usepackage{qcircuit}
\usepackage{dsfont}
\usepackage{hyperref}

\def\identity{\leavevmode\hbox{\small1\kern-3.8pt\normalsize1}}

\newtheorem{theorem}{Theorem}

\newtheorem{lemma}{Lemma}
\newtheorem{proposition}{Proposition}

\newcommand{\spa}{\operatorname{span}}
\newcommand{\mmod}{\!\mod}

\newcommand{\GHZ}{\mathrm{GHZ}}

\newcommand{\Stab}{\mathrm{Stab}}

\newcommand{\Var}{\mathrm{Var}}
\newcommand{\Sp}{\mathrm{Sp}}

\newcommand{\sh}{\mathrm{sh}}

\newcommand{\Cl}{\mathrm{Cl}}

\newcommand{\Ob}{\mathfrak{O}}

\newcommand{\len}{\mathrm{len}}

\newcommand{\Echz}{\textit{Echelon-z}}
\newcommand{\Echx}{\textit{Echelon-x}}

\newcommand{\REchz}{\textit{R-Echelon-z}}
\newcommand{\REchx}{\textit{R-Echelon-x}}

\newcommand{\cutz}{\textit{cut-z}}
\newcommand{\cutx}{\textit{cut-x}}

\newcommand{\bbZ}{\mathbb{Z}}

\newcommand{\bbF}{\mathbb{F}}
\newcommand{\bbI}{\mathbb{I}}
\newcommand{\bbE}{\mathbb{E}}

\newcommand{\bfb}{\mathbf{b}}
\newcommand{\bfe}{\mathbf{e}}
\newcommand{\bfg}{\mathbf{g}}
\newcommand{\bfh}{\mathbf{h}}

\newcommand{\bfm}{\mathbf{m}}

\newcommand{\bfq}{\mathbf{q}}

\newcommand{\bfu}{\mathbf{u}}
\newcommand{\bfv}{\mathbf{v}}
\newcommand{\bfw}{\mathbf{w}}
\newcommand{\bfs}{\mathbf{s}}
\newcommand{\bfx}{\mathbf{x}}
\newcommand{\bfy}{\mathbf{y}}

\newcommand{\bmS}{\bm{S}}

\newcommand{\caC}{\mathcal{C}}
\newcommand{\caD}{\mathcal{D}}
\newcommand{\caE}{\mathcal{E}}
\newcommand{\caG}{\mathcal{G}}
\newcommand{\caH}{\mathcal{H}}
\newcommand{\caJ}{\mathcal{J}}
\newcommand{\caL}{\mathcal{L}}
\newcommand{\caM}{\mathcal{M}}
\newcommand{\caO}{\mathcal{O}}

\newcommand{\caS}{\mathcal{S}}

\newcommand{\caV}{\mathcal{V}}
\newcommand{\caW}{\mathcal{W}}

\newcommand{\rmi}{\mathrm{i}}

\newcommand{\rme}{\operatorname{e}}

\newcommand{\rmU}{\mathrm{U}}

\newcommand{\scrV}{\mathscr{V}}

\newcommand{\bQ}{\bar{Q}}

\newcommand{\ho}{\hat{o}}
\newcommand{\hrho}{\hat{\rho}}

\newcommand{\tbfx}{\tilde{\mathbf{x}}}
\newcommand{\tcaG}{\tilde{\mathcal{G}}}
\newcommand{\tcaJ}{\tilde{\mathcal{J}}}
\newcommand{\tOb}{\tilde{\Ob}}
\newcommand{\tomega}{\tilde{\omega}}

\renewcommand{\epsilon}{\varepsilon}

\def\eqref#1{\textup{(\ref{#1})}}
\newcommand{\eref}[1]{Eq.~\textup{(\ref{#1})}}

\newcommand{\lref}[1]{Lemma~\ref{#1}}

\newcommand{\Lref}[1]{Lemma~\ref{#1}}

\newcommand{\thref}[1]{Theorem~\ref{#1}}
\newcommand{\thsref}[1]{Theorems~\ref{#1}}
\newcommand{\Thref}[1]{Theorem~\ref{#1}}

\newcommand{\eqsref}[2]{Eqs.~(\ref{#1}) and (\ref{#2})}
\newcommand{\tref}[1]{Table~\ref{#1}}
\newcommand{\sref}[1]{Sec.~\ref{#1}}
\newcommand{\Tref}[1]{Table~\ref{#1}}

\newcommand{\pref}[1]{Proposition~\ref{#1}}

\newcommand{\Pref}[1]{Proposition~\ref{#1}}

\newcommand{\fref}[1]{Fig.~\ref{#1}}
\newcommand{\Fref}[1]{Figure~\ref{#1}}

\def\<{\langle}  
\def\>{\rangle}  

\newcommand{\rcite}[1]{Ref.~\cite{#1}}
\newcommand{\rscite}[1]{Refs.~\cite{#1}}

\begin{document}
\title{Qudit Shadow Estimation Based on the Clifford Group and the Power of a Single Magic Gate}

\author{Chengsi Mao}

\affiliation{State Key Laboratory of Surface Physics, Department of Physics, and Center for Field Theory and Particle Physics, Fudan University, Shanghai 200433, China}

\affiliation{Institute for Nanoelectronic Devices and Quantum Computing, Fudan University, Shanghai 200433, China}

\affiliation{Shanghai Research Center for Quantum Sciences, Shanghai 201315, China}

\author{Changhao Yi}

\affiliation{State Key Laboratory of Surface Physics, Department of Physics, and Center for Field Theory and Particle Physics, Fudan University, Shanghai 200433, China}

\affiliation{Institute for Nanoelectronic Devices and Quantum Computing, Fudan University, Shanghai 200433, China}

\affiliation{Shanghai Research Center for Quantum Sciences, Shanghai 201315, China}

\author{Huangjun Zhu}

\email{zhuhuangjun@fudan.edu.cn}

\affiliation{State Key Laboratory of Surface Physics, Department of Physics, and Center for Field Theory and Particle Physics, Fudan University, Shanghai 200433, China}

\affiliation{Institute for Nanoelectronic Devices and Quantum Computing, Fudan University, Shanghai 200433, China}

\affiliation{Shanghai Research Center for Quantum Sciences, Shanghai 201315, China}

\begin{abstract}
Shadow estimation is a sample-efficient protocol for learning the properties of a quantum system using randomized measurements, but the current understanding of qudit shadow estimation is quite limited compared with the qubit setting. Here we clarify the sample complexity
of qudit shadow estimation based on the Clifford group, where the local dimension $d$ is an odd prime. Notably, we show that  the overhead of qudit shadow estimation over the qubit counterpart is  only $\caO(d)$, independent of the qudit number $n$, although the set of stabilizer states may deviate exponentially from a 3-design with respect to the third moment operator.
Furthermore, by adding one layer of magic gates, we propose a simple circuit that can significantly boost the efficiency. Actually, a single magic gate can already eliminate the $\caO(d)$ overhead in qudit shadow estimation and bridge the gap from the qubit setting. 
\end{abstract}
\date{\today}
\maketitle

\emph{Introduction}---Learning the properties of an unknown quantum system is of both fundamental and practical importance in quantum science and technology. The classical shadow estimation was proposed as a sample-efficient protocol for fulfilling this task \cite{HuangKP20,Huan22,KlieR21,ElbeFHK23}. Not requiring a full classical description of the quantum state, this protocol circumvents the `curse of dimensionality' of traditional quantum  tomography \cite{HaahHJW16,CramPFS10}.
Recently, a number of variants have been proposed to make the protocol more robust and practical on noisy
intermediate-scale quantum (NISQ) devices. For example, noise calibration \cite{ChenYZF21,KohG22} and error-mitigation techniques \cite{JnanSC24,SeifCZC23} can make shadow estimation noise-resilient. 
Alternative schemes based on shallow circuits \cite{BertHHI22,IppoLRK23,SchuHH24}, locally scrambled unitary ensembles \cite{HuCY21,TranMHC23}, generalized measurements \cite{NguyBSG22,InnoLPA23}, and hybrid framework~\cite{ZhouL24} are also appealing to practical applications. 
The efficiency of this approach has been demonstrated in a number of experiments \cite{StruZKS21,ZhanSFZ21,StriMPE22}. 

The Clifford group is one of the most important groups in quantum information processing and plays a pivotal role in quantum computation, error correction, and randomized benchmarking \cite{Gott97,NielC10,EiseHWR20,KlieR21}.
In the case of qubits, it is particularly appealing to shadow estimation because it forms a unitary 3-design \cite{Zhu17,Webb16} and is thus maximally efficient in important tasks such as fidelity estimation \cite{HuangKP20}. Recently, estimation protocols based on the Clifford group or its variants have attracted increasing attention and found extensive applications in various estimation problems \cite{HelsW22,BertHHI22,IppoLRK23,SchuHH24,HuanTFS23,LevyLC24,HuanCP23}.

Nevertheless, physical systems as information carriers are not necessarily binary, but typically exhibit multilevel structures, which can be exploited as a resource \cite{WeitAR11,RingMPS22}. Moreover, such systems may help reduce circuit complexity \cite{WangHSK20,ChuHZY23}, enhance efficiency in quantum error correction \cite{Camp14} and magic state distillation \cite{CampAB12}, and improve noise robustness in Bell and steering tests \cite{CollGLMP02,VertPB10,SrivVM22} and  quantum cryptography~\cite{CerfBKG02}.
Meanwhile, developments on the experimental control of qudit systems with photonics, solid-state, trapped ion, and superconducting platforms have made universal and programmable qudit-based quantum processors possible \cite{ChiHZM22,ErhaKZ20,ChoiCLK17,CervKAG22}. Understanding the efficiency of qudit shadow estimation is thus of key theoretical and practical interest. However,  the qudit Clifford group is not a unitary 3-design  \cite{Zhu17,Webb16,KuenG15}, which means it is substantially more difficult to understand its performance in shadow estimation. Actually, little is known about this issue so far despite the intensive efforts of many researchers.

In this Letter, we clarify the sample complexity of qudit shadow estimation based on the Clifford group, where the local dimension $d$ is an odd prime. We show that  the overhead of qudit shadow estimation over the qubit counterpart is only $\caO(d)$, irrespective of the system size $n$.
Furthermore, by adding one layer of magic gates, we propose a simple circuit that can significantly boost the efficiency in shadow estimation. Surprisingly, a single magic gate can already eliminate the $\caO(d)$ overhead  and bridge the gap between qudit systems and qubit systems. Notably, the sample complexity of fidelity estimation is independent of the local dimension and qudit number. To corroborate these theoretical findings, we develop a complete simulation algorithm on qudit shadow estimation based on the Clifford group supplemented by magic gates.  Our work shows that the distinction between the qudit Clifford group and qubit Clifford group is not so serious to quantum information processing. 
This work also highlights the power of a single magic gate in a practical quantum information task, which has not been fully appreciated before. 
  
Our main results are based on a deep understanding of the third moments of Clifford orbits, including the orbits of stabilizer states and magic states in particular. The main proofs of \thsref{thm:GlobalCl}-\ref{thm:Clifford+T} and relevant  technical details are presented in our companion paper \cite{ZhuMY24}, which builds on the previous work~\cite{GrosNW21}. 

\emph{Shadow estimation}---Suppose the local dimension $d$ is a prime, and the $n$-qudit Hilbert space has dimension $D=d^n$. Our task is to estimate the expectation values of certain linear operators on $\caH_d^{\otimes n}$ with respect to an unknown $n$-qudit quantum state~$\rho$. To this end, we can apply a random unitary $U$ sampled from a preselected ensemble $\caE$ to rotate the state ($\rho \mapsto U\rho U^{\dag}$), followed by a computational-basis measurement with outcome $\bfb  \in \bbF_d^n$, where $\bbF_d$ is the finite field composed of $d$ elements. This procedure defines a quantum channel  $\caM(\rho):=\bbE\left[U^{\dag} |\bfb\rangle \langle \bfb|U\right]$, and the inverse $\caM^{-1}$ is called the reconstruction map. By virtue of this map we can construct an estimator $\hat{\rho}:=\caM^{-1}\left(U^{\dag} |\bfb\rangle \langle \bfb|U\right)$, called a (classical) shadow of $\rho$, in each run \cite{HuangKP20}.  

Suppose we want to estimate the expectation value of an operator $\Ob$ on $\caH_d^{\otimes n}$. If $N$ samples are available, then  an unbiased estimator $\hat{o}$ for $o:=\tr(\Ob \rho)$ can be constructed from the empirical mean,
\begin{equation}\label{eq:means}
\ho =\frac{1}{N}\sum_{j=1}^{N}\ho_j= \frac{1}{N}\sum_{j=1}^{N} \tr(\Ob \hrho_j),
\end{equation}
where $\hrho_j$ is the estimator for $\rho$ in the $j$th run. The mean square error (MSE) of $\ho$ only depends on the traceless part $\Ob_0:=\Ob-\tr(\Ob)\bbI/D$ of $\Ob$ \cite{HuangKP20}, where $\bbI$ is the identity operator on $\caH_d^{\otimes n}$. It can be upper bounded as follows (see \sref{sup:empiricalMeans} in Supplemental Material (SM) \cite{supp}): \nocite{Gros06,Foll89,Weil64,AndeB06,Nest08,Berg21,BravM21,BravBCC19,LuksRW97}
\begin{equation}\label{eq:meabudget_new}
\<\epsilon^2\> \leq \frac{\|\Ob_0\|^2_\sh}{N},
\end{equation}
where the (squared) shadow norm $\| \Ob \|_\sh^2$ reads \cite{HuangKP20}
\begin{equation}\label{eq:sndef}
\max_{\sigma}  \bbE_{U\sim \caE} \sum_{\bfb\in \bbF_d^n}  \<\bfb|U\sigma U^\dag |\bfb\> \left|\<\bfb|U\caM^{-1}(\Ob)U^\dag |\bfb\>\right|^2.
\end{equation}
The maximization is over all quantum states on $\caH_d^{\otimes n}$. Alternatively, we can use the median of means to suppress the probability of large deviation \cite{HuangKP20}.
In any case, the shadow norm plays a central role in the analysis of the sample complexity and is widely used as the key figure of merit for evaluating the performance of a shadow estimation protocol.

When the ensemble $\caE$ of unitaries is a 2-design, the associated reconstruction map takes on a simple form, $\caM^{-1}(\Ob)=(D+1)\Ob-\tr (\Ob) \bbI$
for any linear operator $\Ob$ acting on $\caH_d^{\otimes n}$. If in addition  $\caE$ is a 3-design, then  $\| \Ob_0 \|_\sh^2 \leq 3 \| \Ob_0 \|_2^2$ \cite{HuangKP20}. Throughout this Letter, we  use $\| \Ob \|_2$ to denote the Hilbert-Schmidt norm (or $2$-norm) and $\| \Ob \|$ to denote the operator norm (or $\infty$-norm) of $\Ob$.

\emph{Shadow estimation based on local Clifford measurements}---First, we choose the local Clifford group for shadow estimation, that is, $\caE = \Cl(d)^{\otimes n}$, where $\Cl(d)$ is the single-qudit Clifford group
(See Appendix A). This is the qudit generalization of the scheme based on `random Pauli measurements'  \cite{HuangKP20}. The associated reconstruction map reads
\begin{equation}\label{eq:reconlocal}
	\caM^{-1}\left(U^{\dag}|\bfb\>\<\bfb|U \right)=\bigotimes_{j=1}^n \bigl[(d+1)U_j^{\dag}|b_j\>\<b_j|U_j - I\bigr],
\end{equation}
where $U=\bigotimes_{j=1}^n U_j$ with  $U_j \in \Cl(d)$ and $I$ is the identity operator on $\caH_d$. The basic properties of the shadow norm are summarized in 
\pref{pro:LocalCl} and \thref{thm:LocalCl} below; see SM \sref{sup:LocalCl} for proofs. \Pref{pro:LocalCl} also follows from Theorem 4 in the companion paper \cite{ZhuMY24}. As in the qubit case, this strategy is efficient if $\Ob$ only acts on a few qudits. 

\begin{proposition}\label{pro:LocalCl}
Suppose $\Ob$ is an $m$-local Weyl operator with $m\geq 1$. Then its shadow norm associated with local Clifford measurements reads 
	\begin{equation}
	\| \Ob \|_\sh^2=(d+1)^m.
	\end{equation}
\end{proposition}

\begin{theorem}\label{thm:LocalCl}
	Suppose  $\Ob$ is an $m$-local linear operator on $\caH_d^{\otimes n}$, that is, $\Ob=\tOb \otimes I^{\otimes (n-m)}$ with $\tOb$ acting on $\caH_d^{\otimes m}$. Then its shadow norm associated with local Clifford measurements satisfies
	\begin{equation}\label{eq:local linear}
		\| \Ob \|_\sh^2 \leq d^m  \| \tOb \|_2^2.
	\end{equation}
\end{theorem}

\emph{Shadow estimation based on global Clifford measurements}---Next, we turn to shadow estimation based on the global Clifford group, i.e., $\caE = \Cl(n,d)$ (See Appendix A). Equivalently, the associated measurement primitive is the orbit $\Stab(n,d)$ of all $n$-qudit stabilizer states. Since the Clifford group $\Cl(n,d)$ forms a unitary 2-design, any Clifford orbit forms a 2-design,
and the corresponding reconstruction map is particularly simple,
\begin{equation}\label{eq:reconstruction}
\caM^{-1}(U^{\dag}|\bfb\>\<\bfb|U)=(D+1)U^{\dag}|\bfb\>\<\bfb|U-\bbI.
\end{equation}
However, $\Stab(n,d)$ does not form a 3-design when $d$ is an odd prime \cite{KuenG15}, and it is substantially more difficult to determine the shadow norm of a generic observable. In the companion paper \cite{ZhuMY24} we have clarified the properties of the third moment of $\Stab(n,d)$, which enables us to derive the following universal upper bounds for the shadow norm. A sketch of the proof is presented in SM \sref{sup:GlobalClProofs}.

\begin{theorem}\label{thm:GlobalCl}
	Suppose  $\caE = \Cl(n,d)$ and $\Ob$ is a traceless linear operator on $\caH_d^{\otimes n}$. Then 
	\begin{equation}\label{eq:GlobalLinear}
	\|\Ob\|_2^2 \leq \|\Ob\|^2_\sh \leq (2d-3)\|\Ob\|_2^2+2\|\Ob\|^2.
	\end{equation}
	If  $\Ob$ is diagonal in a stabilizer basis, then
	\begin{equation}\label{eq:GlobalStab}
	\|\Ob\|^2_\sh \leq (d-1)\|\Ob\|_2^2+d\|\Ob\|^2.
	\end{equation}
	If  $n = 1$ and $\Ob$ is diagonal in a stabilizer basis, then
	\begin{equation}\label{eq:GlobalStabn1}
	\|\Ob\|^2_\sh = (d+1)\|\Ob\|^2.
	\end{equation}
\end{theorem} 
Incidentally,  \eref{eq:GlobalStabn1} implies \pref{pro:LocalCl}. Thanks to \thref{thm:GlobalCl}, the ratio $
\|\Ob\|^2_\sh/ \|\Ob\|_2^2$ is upper bounded by $2d-1$, independent of the qudit number $n$, so the overhead of shadow estimation for a qudit system over a qubit system does not grow with $n$.  This result may have profound implications for quantum information processing based on qudits, and is quite unexpected given that $\Stab(n,d)$ is not a 3-design. In fact, the operator norm of the third normalized moment operator of $\Stab(n,d)$ increases exponentially with $n$ when $d\neq 2\mmod 3$~\cite{ZhuMY24}. Surprisingly, measurement ensembles far from 3-designs (with respect to one of the most popular figures of merit) can achieve similar performances as 3-designs in shadow estimation. When $n=1$, $\Stab(n,d)$ forms a complete set of mutually unbiased bases, so \eref{eq:GlobalStabn1} is also instructive to understanding shadow estimation based on  mutually unbiased bases. 

\begin{figure}[tb]
	\centering 
	\includegraphics[width=0.48\textwidth]{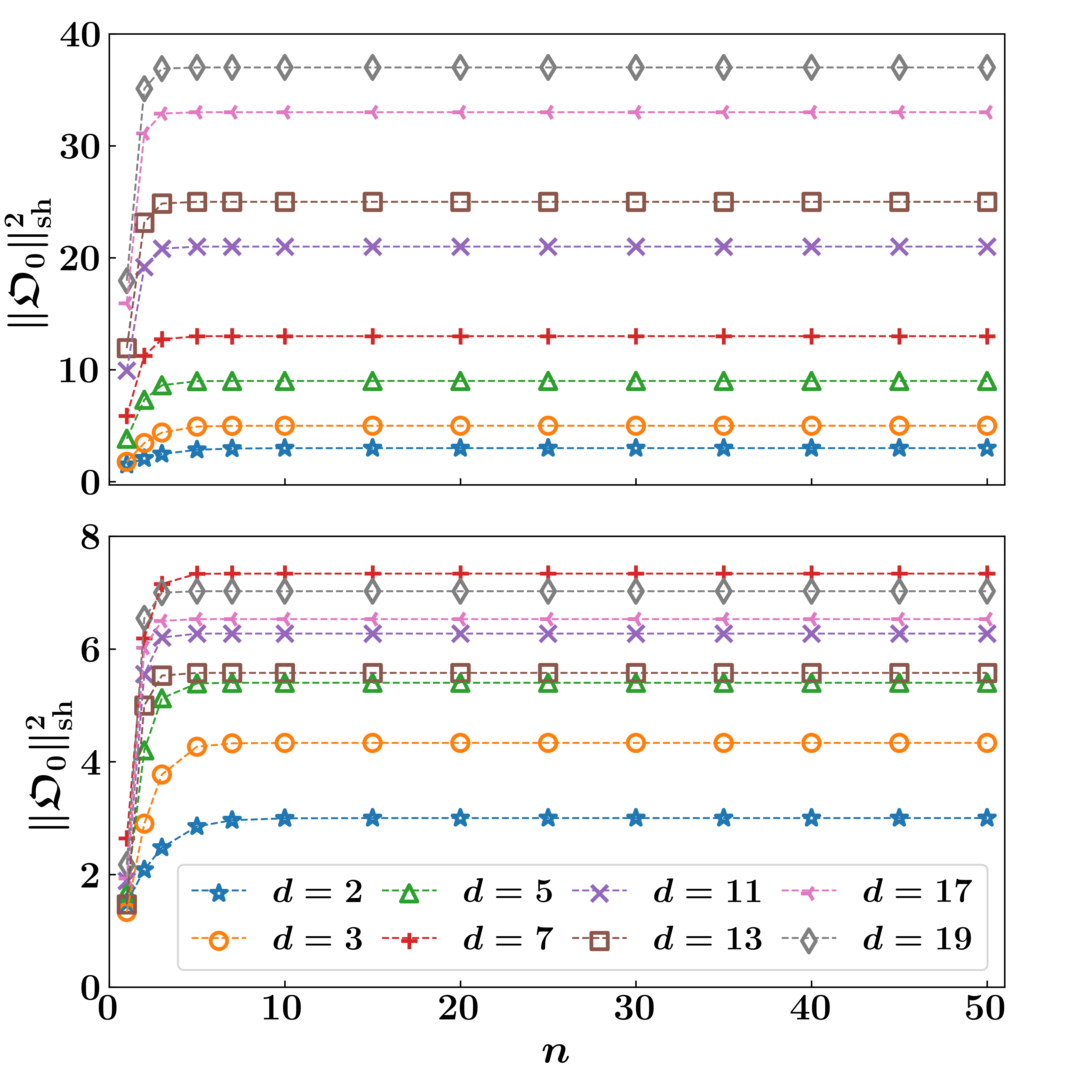}
	\caption{Shadow norm associated with the projector $\Ob$ onto a stabilizer state with respect to the Clifford circuit (upper plot) and Clifford circuit  supplemented by one canonical $T$ gate (lower plot). Here $\Ob_0$ denotes the traceless part of $\Ob$.}
	\label{fig:stabstatenorm1}
\end{figure}

When $\Ob$ is a stabilizer projector, we can even derive an analytical formula for $\|\Ob_0\|^2_\sh$.    
\begin{theorem}\label{thm:GlobalStabProj}If $\caE = \Cl(n,d)$ and $\Ob$
 is a rank-$K$ stabilizer projector on $\caH_d^{\otimes n}$ with $1\leq K\leq d^{n-1}$, then 
	\begin{equation}\label{eq:StabProj}
	\frac{\|\Ob_0\|^2_\sh}{\|\Ob_0\|^2_2}= \frac{D+1}{D+d}\left(d-1-\frac{d}{D}+\frac{d}{K}\right).
	\end{equation}
\end{theorem}

The shadow norm $\|\Ob_0\|^2_\sh$ increases linearly with the local dimension $d$, but is almost independent of $n$ when $n\geq 5$, as illustrated in \fref{fig:stabstatenorm1} for the case $K=1$, so that $\Ob$ is the projector onto a stabilizer state. In addition, $\|\Ob_0\|^2_\sh$ can saturate the upper bounds in \eqsref{eq:GlobalLinear}{eq:GlobalStab} asymptotically  as $n \to \infty$. Interestingly, stabilizer states are the most difficult to estimate by stabilizer measurements. 

\begin{figure}[tb]
	\centering
	\includegraphics[width=0.44\textwidth]{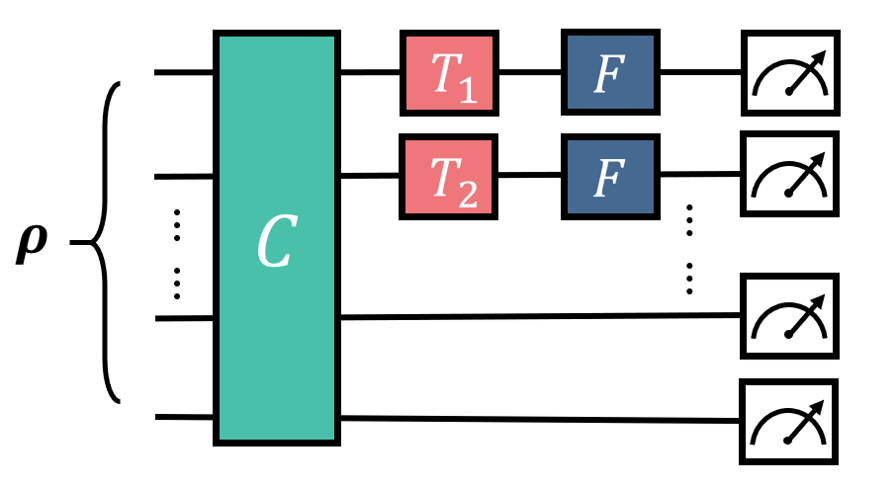}
	\caption{A Clifford circuit supplemented by one layer of $T$ gates, which is highly efficient for shadow estimation.}	
	\label{fig:Cl_magic_circuit}
\end{figure}

\emph{The power of magic gates in shadow estimation}---Clifford circuits supplemented by sufficiently many  magic gates can realize universal quantum computation, but little is known about the power of limited magic gates in quantum information processing.  Here we propose a simple recipe for boosting the efficiency of qudit shadow estimation  and show that a single magic gate can already bridge the gap between a qudit system and a qubit system.

We are mainly interested in qudit magic gates that are diagonal in the computational basis \cite{HowaV12,ZhuMY24} and belong to the third Clifford hierarchy \cite{GottC99}, which are discussed in Appendix B and referred to as $T$ gates henceforth. For example, a canonical $T$ gate has the form  $T=\sum_{b\in\bbF_d} \tomega^{b^3} |b\>\<b|$, where 
$\tomega=\omega_d:=\rme^{\frac{2\pi \rmi}{d}}$ for $d\geq 5$ and  $\tomega:=\rme^{\frac{2\pi \rmi}{9}}$ for $d=3$. Using $T$ gates we can construct an alternative unitary ensemble for shadow estimation as illustrated in \fref{fig:Cl_magic_circuit}: First apply a random Clifford unitary $C$ in  $\Cl(n,d)$, and then apply  $T$ gates  on $k$ different  qudits (the first $k$ qudits for simplicity), each followed by a Fourier gate $F:=\frac{1}{\sqrt{d}} \sum_{a,b \in \bbF_d} \omega_d^{ab} |a\>\<b|$.
Denote by $T_j$ the $T$ gate applied to the $j$th qudit and by $\caE\left(\{T_j\}_{j=1}^k\right)$
the resulting unitary ensemble. The effective measurement ensemble corresponds to the Clifford orbit generated from a magic state. 
The circuit we employ is substantially simpler than the popular interleaved Clifford circuits: in each run 
we need to sample a random Clifford unitary only once, which is as economical as possible and is particularly appealing in the NISQ era \cite{Pres18}. Nevertheless, such simple circuits are surprisingly  powerful in shadow estimation as shown in the following theorem; see SM \sref{sup:GlobalClProofs} for a proof sketch.

\begin{theorem}\label{thm:Clifford+T}
	Suppose $\caE=\caE\left(\{T_j\}_{j=1}^k\right)$ with $1 \leq k\leq n$, where $T_1, T_2, \ldots, T_k$ are $T$ gates, and	$\Ob$ is a traceless linear operator on $\caH_d^{\otimes n}$. Then 
	\begin{equation}\label{eq:C+T}
	\|\Ob\|_\sh^2 \leq \tilde{\gamma}_{d,k} \|\Ob\|_2^2,
	\end{equation}
	where
	\begin{equation}\label{eq:gamma}
	\tilde{\gamma}_{d,k} := 
	\begin{cases}
	3+\frac{2^{k+1}(d-2)}{d^k} & d\neq 1 \mmod 3, \\
	3+\frac{9}{8}\cdot\frac{4^k}{d^{k-1}} & d=1 \mmod 3.
	\end{cases}
	\end{equation}
\end{theorem}

\begin{figure}
	\centering 
	\includegraphics[width=0.48\textwidth]{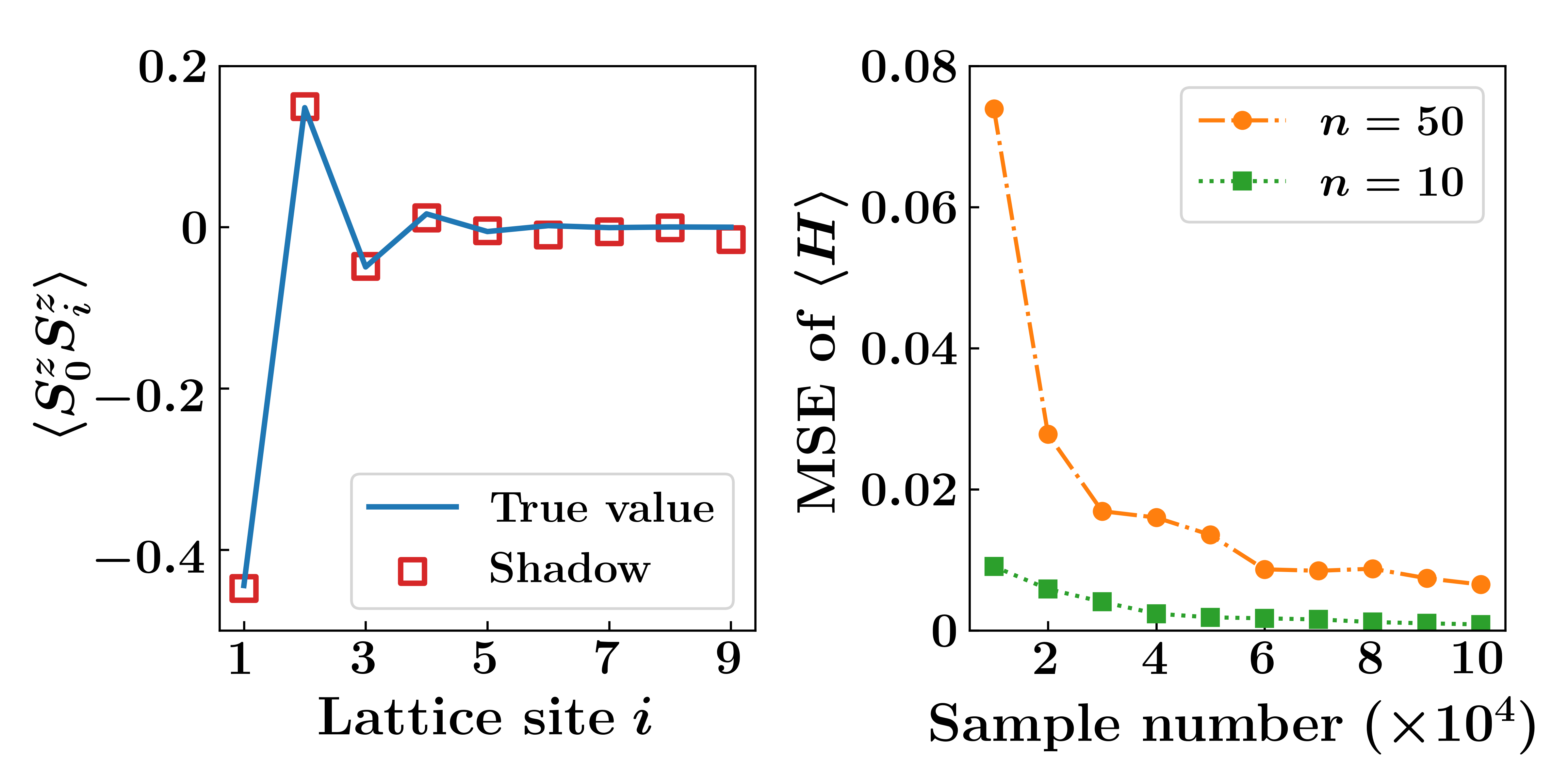}
	\caption{(left) Estimation of the two-point correlation function $\< S^z_0 S^z_i \>$ for the ground state of a 1D AKLT chain with $10$  sites, based on $10^5$ random local Clifford measurements. (right) The MSE in estimating the ground energy $\< H \>$, averaged over $100$ independent runs.}
	\label{fig:aklt}
\end{figure}

\begin{figure*}[bt]
	\centering 
	\includegraphics[width=0.8\textwidth]{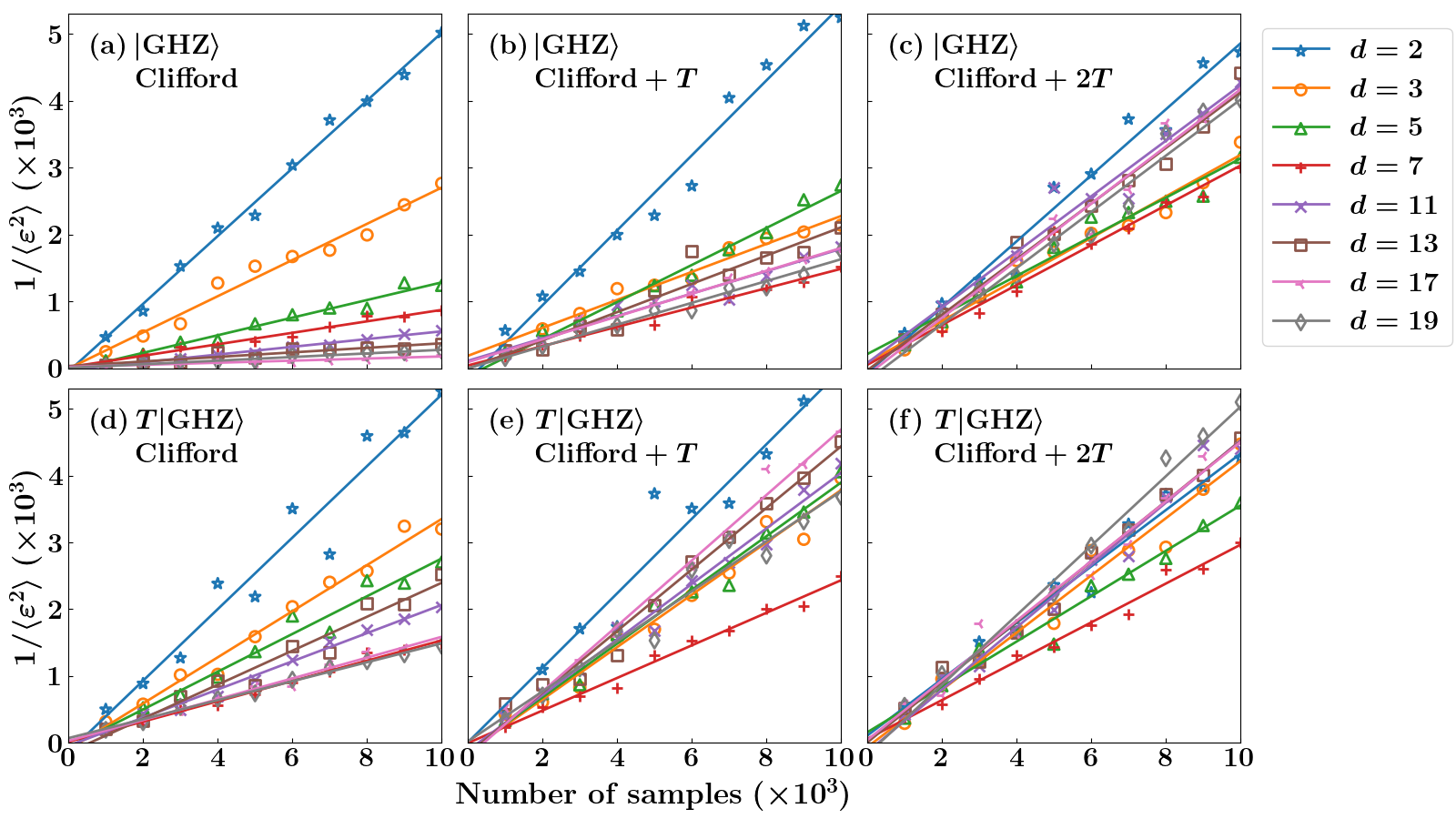}
	\caption{Simulation results on the inverse MSE in fidelity estimation in which the input state and target state are identical $n$-qudit GHZ states, $|\GHZ\>=|\GHZ(n,d)\>$, with $n=100$. 	
		The estimation protocols are  based on the Clifford circuit supplemented by up to two canonical $T$ gates. The MSE $\<\epsilon^2\>$ for each data point is the average over 100 runs. The solid lines are determined by interpolation. Results on the state $T |\GHZ\>$ are also shown for comparison.}
	\label{fig:GHZ}
\end{figure*}

As the number $k$ of $T$ gates increases, the shadow norm associated with the  ensemble $\caE\left(\{T_j\}_{j=1}^k\right)$ converges exponentially to the counterpart of a unitary 3-design. Moreover, by adding a single $T$ gate, the ratio $\|\Ob\|_\sh^2 / \|\Ob\|_2^2$ can already be upper bounded by a constant. Notably, if the observable $\Ob$ has a bounded Hilbert–Schmidt norm, which is the case for many tasks, such as fidelity estimation, then the sample complexity is independent of the local dimension $d$ and qudit number $n$. When $\Ob$ is the projector onto a stabilizer state, we can  
extend \thref{thm:GlobalStabProj} and derive an analytical formula for $\|\Ob_0\|_\sh^2$ as shown in our companion paper \cite{ZhuMY24} and illustrated in \fref{fig:stabstatenorm1}. \Thref{thm:Clifford+T} is applicable irrespective of the specific choices of $T$ gates. If we can make educated choices, then the shadow norm $\|\Ob_0\|_\sh$ can be reduced further as shown in SM \sref{sup:Numerical}.

\emph{Numerical simulation}---First, we use local Clifford measurements to explore several properties of the 1D spin-1 Affleck–Kennedy–Lieb–Tasaki (AKLT) model with open boundaries \cite{AfflKLT87,AfflKLT88}, referred to as  `1D AKLT chain' henceforth. This model plays a key role in condensed matter physics and measurement-based quantum computation. The Hamiltonian reads
\begin{equation}
	H = \frac{1}{2} \sum_j  \left[\bmS_j \cdot \bmS_{j+1} + \frac{1}{3} \bigl(\bmS_j \cdot \bmS_{j+1}\bigr)^2 + \frac{2}{3}\bbI\right].
\end{equation}
Note that the spin on each lattice site can be regarded  as a qutrit ($d=3$).  
In the left plot of \fref{fig:aklt}, we show the two-point correlation function $\langle S^z_0 S^z_i \rangle$ on a 1D AKLT chain with $n=10$ spins estimated from $10^5$ samples. Here the estimator based on local Clifford measurements is constructed from the empirical mean as in \eref{eq:means}. The result is consistent with the theoretical prediction $\langle S^z_0 S^z_i \rangle = \frac{4}{3} (-3)^{-i}$ \cite{AfflKLT87}. In the right plot we show the MSE in estimating the ground energy, which  converges faster for a smaller system size $n$, as expected from \thref{thm:LocalCl}.

Next, we turn to qudit shadow estimation based on the Clifford circuit supplemented by a layer of $T$ gates. As a showcase, we consider the task of fidelity estimation in which the input state $\rho$ is identical to a pure target state $|\Psi\>\<\Psi|$. In this case, the observable of interest is $\Ob=|\Psi\>\<\Psi|$ and the true fidelity is $\tr(\rho \Ob) = 1$. We first test our protocols on the $n$-qudit Greenberger-Horne-Zeilinger (GHZ) state, that is,  $|\Psi\>=|\GHZ(n,d)\> $ with
\begin{equation}\label{eq:GHZ}
|\GHZ(n,d)\> := \frac{1}{\sqrt{d}} \sum_{b =0}^{d-1}|b \>^{\otimes n}.
\end{equation}
In numerical simulation we set $n=100$, but the number of qudits has little influence on the sample complexity. \Fref{fig:GHZ} shows the simulation results on the inverse MSE $1/\<\epsilon^2\>$, which increases  linearly with the number of samples as expected from \eref{eq:meabudget_new}. The smaller the shadow norm, the larger the slope of the interpolation line. When applying Clifford circuits without $T$ gates for shadow estimation, the slope is approximately inversely proportional to the local dimension $d$, but the dependence becomes negligible as more $T$ gates are applied, which confirms our theoretical prediction.  

For comparison, we also test our protocols on the state $T |\GHZ\>$, where $T$ is the canonical $T$ gate acting on the first qudit. Now the MSE is smaller than the $|\GHZ\>$ counterpart and we see a `duality' between magic gates in state preparation and in measurements: the MSE is mainly determined by the total number of $T$ gates. Additional simulation results on depolarized GHZ states,  cluster states, and median-of-means estimation can be found in Appendix~D and SM; the general conclusions are similar.

\emph{Summary}---We showed that qudit shadow estimation based on the Clifford group can achieve the same precision as in the qubit setting with only $\caO(d)$ overhead in the sample complexity, irrespective of the qudit number $n$. 
Furthermore, we proposed a simple recipe to boost the efficiency, which requires only one layer of magic gates in addition to the Clifford circuit and is particularly appealing to the NISQ era. Actually, by adding a single magic gate, we can  eliminate the $\caO(d)$ overhead in qudit shadow estimation and achieve the same sample complexity as in the qubit setting. In this way, we can achieve a constant sample complexity (independent of $d$ and $n$) in important tasks,  such as fidelity estimation.
These results means characterization and verification of qudit systems are easier than expected, which may have profound implications for quantum information processing based on qudits. Meanwhile, our work highlights the power of a single magic gate in a practical quantum information processing task, which is of independent interest and deserves further exploration.

\emph{Acknowledgments}---This work is supported by Shanghai Science and Technology Innovation Action Plan (Grant No.~24LZ1400200), Shanghai Municipal Science and Technology Major Project (Grant No.~2019SHZDZX01), National Key Research and Development Program of China (Grant No.~2022YFA1404204), and National Natural Science Foundation of China (Grant No.~92165109).

\let\oldaddcontentsline\addcontentsline
\renewcommand{\addcontentsline}[3]{}

\bibliography{ref}

\let\addcontentsline\oldaddcontentsline

\onecolumngrid	
\vspace{1cm}
\begin{center}
\textbf{\large End Matter}
\end{center}
\twocolumngrid	

\setcounter{equation}{0}
\renewcommand{\theequation}{A\arabic{equation}}
\emph{Appendix A: Heisenberg-Weyl group and Clifford group}---The phase operator $Z$ and cyclic-shift operator $X$ for a single qudit are defined as follows,
\begin{equation}\label{eq:ZX}
	Z|b\> = \omega_d^b|b\>, \quad X|b\> = |b+1\>,
\end{equation}
where $\omega_d=\rme^{\frac{2\pi \rmi}{d}}$ and the addition $b+1$ is modulo $d$. When $d$ is an odd prime, the qudit Heisenberg-Weyl (HW) group $\caW(d)$ is generated by $Z$ and $X$; when $d=2$, the HW group $\caW(d)$ is generated by $Z$, $X$, $\rmi I$ and reduces to the Pauli group.  The $n$-qudit HW group $\caW(n,d)$ is the tensor product of $n$ copies of $\caW(d)$, and its elements are called Weyl operators. A Weyl operator is trivial if it is proportional to the identity and nontrivial otherwise. It is $m$-local if it is a tensor product of $m$ nontrivial single-qudit Weyl operators and $(n-m)$ identity operators on $\caH_d$.

The single-qudit Clifford group $\Cl(d)$ is the normalizer of the HW group $\caW(d)$. The local Clifford group $\Cl(d)^{\otimes n}$ is the tensor product of $n$ copies of $\Cl(d)$. By contrast, the (global) Clifford group $\Cl(n,d)$ is the normalizer of $\caW(n,d)$. 
When $d=2$, the Clifford group $\Cl(n,d)$ is a 3-design \cite{Zhu17,Webb16},  and this is the only known infinite family of 3-designs based on discrete groups (up to global phase factors). Indeed, the performance guarantee of qubit shadow estimation is rooted in the 3-design property of $\Cl(n,2)$. When $d$ is an odd prime, by contrast, the Clifford group  $\Cl(n,d)$ is only a 2-design \cite{Zhu17,Webb16}, and little is known about the efficiency of qudit shadow estimation based on the Clifford group, which motivates the current study. 

\setcounter{equation}{0}
\renewcommand{\theequation}{B\arabic{equation}}
\emph{Appendix B: Magic gates}---Up to an irrelevant global phase factor, a qudit $T$ gate takes on the following form \cite{HowaV12,ZhuMY24},
\begin{equation}\label{eq:T_f}
	T:=\sum_{b \in \bbF_d} \tomega^{f(b )}|b \>\<b |, \quad
	\tomega:=
	\begin{cases}
		\omega_d \quad d \geq 5,\\
		\omega_9\quad d=3.
	\end{cases}
\end{equation}
When $d\geq 5$, $f$ is a cubic polynomial on $\bbF_d$; when $d=3$, $f(b )=c_3 b ^3 +3c_2 b ^2$ (with $c_2\in \bbF_3$, $c_3 \in \bbZ_9$ and $c_3\neq 0\mmod 3$) is a function from $\bbF_3$ to $\bbZ_9$.
The $T$ gate associated with the function $f(b )=b ^3$ is  referred to as the canonical $T$ gate.

When $d = 1 \!\mod 3$, each $T$ gate is determined by a cubic polynomial $f$ over $\bbF_d$. We can distinguish three types of $T$ gates depending on the cubic coefficient~$c$  ($c\neq 0$ by definition) of the underlying cubic polynomial $f$ \cite{ZhuMY24}. More precisely, $T$ gates are distinguished by cubic characters of cubic coefficients. Let $\nu$ be a given primitive element of the field $\bbF_d$; see \tref{tab:nu} in SM \sref{sup:Numerical} for a specific choice when $d<50$. Let $\eta_3$ be the 
cubic character on $\bbF_d^\times:=\bbF_d\setminus \{0\}$ defined as follows,
\begin{equation}\label{eq:eta3}
	\eta_3(\nu^j) := \omega_3^j, \quad j=0,1,\ldots, d-2, 
\end{equation}
where $\omega_3= \rme^{2\pi \rmi/3}$. 
The cubic character of $T$ is defined as the cubic character of the cubic coefficient~$c$, that is, $\eta_3(T)= \eta_3(c)$. Note that $\eta_3(T)=1$ iff the cubic coefficient $c$ is a cubic residue,  that is, a cube of another element in $\bbF_d^\times$. For example, we have $\eta_3(T)=1$ when $c=1$, which is the case for the canonical $T$ gate. 
It turns out different types of $T$ gates have different performances in shadow estimation; see SM \sref{sup:Numerical} and the companion paper \cite{ZhuMY24} for more details.  In other words, we can improve the performance by educated choices of $T$ gates. Incidentally, when  $d=3$ or $d=2\!\mod 3$, all $T$ gates have the same performance, so it is not necessary to distinguish different $T$ gates. 
 
\begin{figure}[t]
	\centering 
	\includegraphics[width=0.48\textwidth]{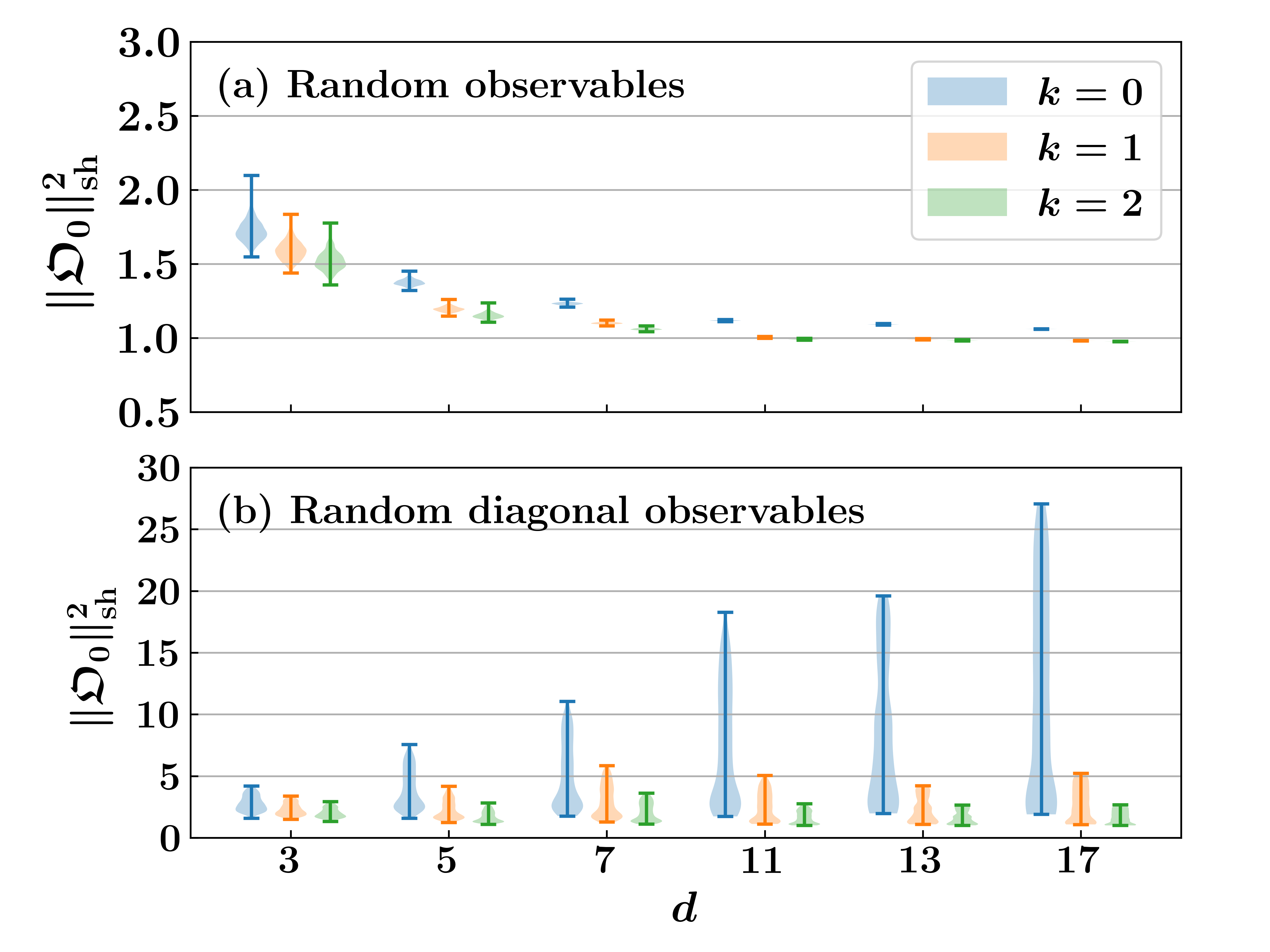}
	\caption{Distributions of the shadow norms of random two-qudit traceless  observables (a) and random two-qudit traceless diagonal  observables (b) that are normalized with respect to the Hilbert-Schmidt norm. The unitary ensemble $\caE$ is based on the Clifford circuit supplemented by $k$ canonical $T$ gates. 
	}
	\label{fig:shadownorm}
\end{figure}

\emph{Appendix C: Distributions of shadow norms of generic observables}---To show how tight the upper bounds in \thsref{thm:GlobalCl} and \ref{thm:Clifford+T} are for generic observables, we randomly sample 2000 two-qudit observables and two-qudit diagonal observables, respectively, which are traceless and normalized with respect to the Hilbert-Schmidt norm. The results on their shadow norms 
are shown in \fref{fig:shadownorm}. 
For generic observables, the shadow norms are usually much smaller than the upper bound in \eref{eq:GlobalLinear} and do not increase with the local dimension; for diagonal observables, by contrast, the upper bound in \eref{eq:GlobalStab} is nearly tight. 
The addition of $T$ gates can significantly reduce the shadow norms of random diagonal observables, although their utility for generic random observables seems limited because the Clifford circuit is already good enough.

\begin{figure}
	\centering 
	\includegraphics[width=0.48\textwidth]{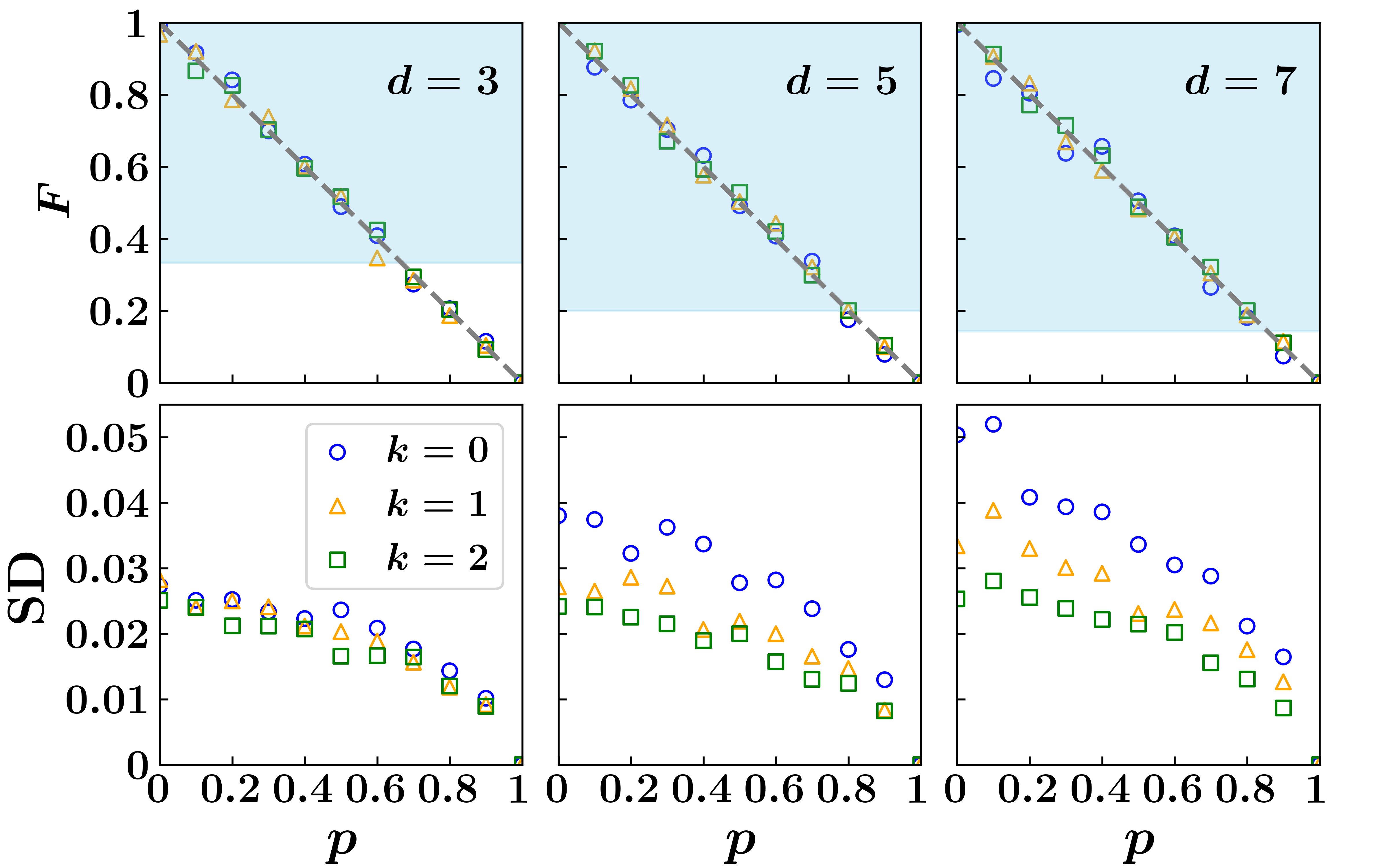}
	\caption{Fidelity estimated using shadow estimation based on global Clifford measurements and the corresponding standard deviation (SD). The target state is the GHZ state in \eref{eq:GHZ} with $n=100$, and the input state is the depolarized GHZ state in Appendix D. The unitary ensemble is based on the Clifford circuit supplemented by $k$ canonical $T$ gates, and the number $N$ of samples is $5000$. Each data point for the fidelity is the average over 100 runs, which also determine the standard deviation.	
	In the upper row, the gray lines show the true fidelity $F=1-p+p/D$, and the light blue areas means the states  are  GME.}
	\label{fig:fide}
\end{figure}

\emph{Appendix D: Entanglement detection via fidelity estimation}---Another important application of shadow estimation is entanglement detection. Suppose the input state is a depolarized GHZ state, that is, $\rho = p\bbI/D+(1-p)|\GHZ(n,d)\>\<\GHZ(n,d)|$, where $0\leq p\leq 1$. A quantum state is genuinely multipartite entangled (GME) if its fidelity with some multipartite entangled state $|\Psi\>$ is larger than the maximum Schmidt coefficient of $|\Psi\>$ maximized over all bipartitions \cite{GuhnT09}. In the case under consideration  we have $|\Psi\>=|\GHZ(n,d)\>$, so the state $\rho$ is GME if $\tr(\rho |\GHZ(n,d)\>\<\GHZ(n,d)|)>1/d$. Using shadow estimation based on the global Clifford group, we can estimate the fidelity between $\rho$ and $|\GHZ(n,d)\>$ and thereby detect genuine multipartite entanglement (GME). \Fref{fig:fide} shows that our protocols based on  $5000$ samples give quite good estimators. Moreover, the fluctuation (standard deviation) gets smaller as more $T$ gates are added to the Clifford circuit. Therefore, the GME of the depolarized GHZ state can be detected efficiently.

\emph{Appendix E: Algorithms for numerical simulation}---For ($T$-modified) stabilizer states, we develop a complete simulation algorithm for qudit shadow estimation based on the Clifford group supplemented by $T$ gates.
When no $T$ gates are involved, the data acquisition phase is essentially a random Clifford circuit. To simulate such a circuit, in SM \sref{sup:Simulation} we generalize the tableau method proposed in \rcite{AaroG04} to the qudit setting and combine it with the Clifford sampling algorithm proposed in \rscite{KoenS14,Hein21}. Since this simulation method is based on the stabilizer formalism, it is highly efficient in terms of both time and space complexities, which enables us to simulate systems consisting of hundreds of qudits.

To deal with the addition of $T$ gates, we then generalize the gadgetization method for insertion of magic gates in the Clifford circuit \cite{BravG16,PashRKB22} to the qudit setting (see SM \sref{sup:Simulation} for more details). Finally, we are able to simulate a single shot in shadow estimation on a classical computer with an approximate computational cost of $\caO\bigl((n+t)^3+td^{t+1}\bigr)$, where $t$ is the number of magic gates involved in state preparation and measurements.

As for the ground states of the AKLT models, since they allow an efficient representation using matrix product states, we adopt the tensor-network paradigm described in \rcite{CarrTMA19} to simulate random local Clifford measurements on 1D AKLT chains.

\bigskip
\emph{Appendix F: Sampling cost of random Clifford unitaries}---In our numerical simulation, we adopt the protocol proposed in \rscite{KoenS14,Hein21} to sample unitaries from $\Cl(n,d)$; the (classical) computational cost is $\caO(n^3)$.
However, in the experiment setting, it is more reasonable to consider sampling cost as the number of elementary gates required to implement a random Clifford unitary. To this end, we may take the elementary gate set presented in \rcite{Hein21} (see also SM \sref{sup:QuditStab}) and generalize the method in \rcite{Berg21} to the qudit case. A quick analysis shows that this procedure can output a random Clifford unitary in the form of a quantum circuit with $\caO(n^2)$ elementary gates, although a small overhead compared with the qubit case is induced due to a higher local dimension.

\clearpage
\newpage

\setcounter{equation}{0}
\setcounter{figure}{0}
\setcounter{table}{0}
\setcounter{theorem}{0}
\setcounter{lemma}{0}
\setcounter{section}{0}
\setcounter{page}{1}

\renewcommand{\theequation}{S\arabic{equation}}
\renewcommand{\thefigure}{S\arabic{figure}}
\renewcommand{\thetable}{S\arabic{table}}
\renewcommand{\thetheorem}{S\arabic{theorem}}
\renewcommand{\thelemma}{S\arabic{lemma}}
\renewcommand{\theproposition}{S\arabic{proposition}}
\renewcommand{\thecorollary}{S\arabic{corollary}}
\renewcommand{\thesection}{S\arabic{section}}

\onecolumngrid	
\begin{center}
	\textbf{\large Qudit Shadow Estimation Based on the Clifford Group and the Power of a Single Magic Gate: Supplemental Material}
\end{center}

\tableofcontents

\bigskip

In this Supplemental Material (SM), we prove \eref{eq:meabudget_new}, \pref{pro:LocalCl}, and \thref{thm:LocalCl} and sketch the proofs of \thsref{thm:GlobalCl}-\ref{thm:Clifford+T} based on our companion paper \cite{ZhuMY24}. In addition, we develop a complete algorithm for simulating qudit shadow estimation based on the Clifford group supplemented by $T$ gates and provide additional numerical results. Background on the qudit stabilizer formalism is also introduced for the convenience of the readers. Throughout this SM the Hilbert space for a single qudit is denoted by $\caH_d$.

\section{Proof of \eref{eq:meabudget_new}}\label{sup:empiricalMeans}
The MSE of the estimator $\hat{o}$ only depends on the traceless part of $\Ob$. To prove \eref{eq:meabudget_new} we can assume that $\Ob$ is traceless without loss of generality; then $\Ob_0=\Ob$.  
Suppose in the $j$th run the random unitary is $U$ and the measurement outcome is $\bfb\in \bbF_d^n$. Then the corresponding single-shot estimator $\hat{o}_j$ reads
\begin{equation}
	\ho_j := \tr\left(\Ob \hrho_j\right) =\tr\left(\Ob_0 \hrho_j\right) = \tr\left[\Ob_0 \caM^{-1}(U^\dag |\bfb\>\<\bfb|U)\right] = \tr\left[\caM^{-1}(\Ob_0) U^\dag |\bfb\>\<\bfb|U \right],
\end{equation}
where the last equality holds because the reconstruction map $\caM^{-1}$ is self-adjoint.
The expectation value of $\hat{o}_j$ reads $\bbE [\ho_j]=o= \tr(\Ob\rho)=\tr(\Ob_0\rho)$. So  the variance of $\ho_j$ can be upper-bounded as follows,
\begin{align}
	\Var[\ho_j] & = \bbE\left[|\ho_j-o|^2\right] \leq \bbE\left[|\ho_j|^2\right] = \bbE\left[\left|\<\bfb|U\caM^{-1}(\Ob_0) U^\dag |\bfb\>\right|^2 \right] \nonumber \\
	& = \bbE_{U \sim \caE} \sum_{\bfb\in \bbF_d^n} \<\bfb|U\rho U^\dag |\bfb\> \left|\<\bfb|U\caM^{-1}(\Ob_0) U^\dag |\bfb\>\right|^2 \leq \|\Ob_0\|_\sh^2,
\end{align}
where the last inequality follows from the definition of the (squared) shadow norm in \eref{eq:sndef}. Note that this derivation is applicable even if $\Ob$ is not Hermitian. 

By definition the empirical mean $\hat{o}$ is the average of $N$ independent single-shot estimators. Therefore,
\begin{equation}
	\<\epsilon^2\>=\Var[\ho] = \frac{1}{N^2} \sum_{j=1}^N \Var[\ho_j]\leq \frac{\|\Ob_0\|^2_\sh}{N},
\end{equation}
which confirms \eref{eq:meabudget_new}.

\section{Qudit stabilizer formalism}\label{sup:QuditStab}
In this section, we provide a bit more details on the qudit stabilizer formalism, assuming that the local dimension $d$ is an odd prime \cite{Gros06}.  

\subsection{Heisenberg-Weyl group and Clifford group}
Recall that the HW group $\caW(d)$ for a single qudit is generated by the two operators $Z$ and $X$ defined in \eref{eq:ZX}, that is, 
\begin{equation}
	\caW(d):=\langle Z,X \rangle = \bigl\{\omega_d^a Z^b X^c \,|\, a,b,c \in \bbF_d\bigr\},
\end{equation}
where  $\omega_d =\rme^{\frac{2\pi \rmi}{d}}$. Up to phase factors, the  elements in $\caW(d)$, known as  Weyl operators,  can be labeled by vectors in $\bbF_d^2$ \cite{Gros06}. Given $u=(p,q)^\top\in \bbF_d^2$, define
\begin{equation}
	W_u=W(p,q) := \chi(-pq/2)Z^{p}X^{q},
\end{equation}
where $\chi(r) = \omega_d^r$ for $r\in \bbF_d$, and the choice of the phase factor $\chi(-pq/2)$  guarantees that $W_u^\dag=W_{-u}$. Then $\caW(d) = \bigl\{\omega_d^a W_u \, |\, a\in \bbF_d, u \in \bbF_d^{2}\bigr\}$. By definition $W_u$ and $W_v$ for $u,v\in \bbF_d^2$ commute with each other iff $u$ and $v$ are linearly dependent. 

The $n$-qudit HW group $\caW(n,d)$ is the tensor product of $n$ copies of $\caW(d)$. Up to phase factors, the elements in $\caW(n,d)$, known as $n$-qudit Weyl operators, can be labeled by column vectors in $\bbF_d^{2n}$. Given $\bfu=(u^z_1, u^z_2, \ldots, u^z_n, u^x_1,u^x_2,\ldots, u^x_n)^\top\in \bbF_d^{2n}$, define
\begin{equation}\label{eq:Wbfu}
	W_\bfu := W_{u_1} \otimes W_{u_2} \otimes \cdots \otimes W_{u_n}, 
\end{equation}
where $u_j=(u^z_j,u^x_j)^\top$  for $j=1,2,\ldots, n$. Then $\caW(n,d) = \bigl\{\omega_d^a W_\bfu \, |\, a\in \bbF_d, \bfu \in \bbF_d^{2n}\bigr\}$. The weight of the Weyl operator $\omega_d^a W_\bfu$ is defined as the number of tensor factors in \eref{eq:Wbfu} that are not proportional to the identity operator $I$; the Weyl operator is $m$-local if it has weight $m$. Given $a,b\in \bbF_d$ and $\bfu,\bfv\in \bbF_d^{2n}$, the two Weyl operators $\omega_d^a W_\bfu$ and $\omega_d^b W_\bfv$  are locally commutative if $W_{u_j}$ and $W_{v_j}$ are commutative for $j=1,2,\ldots, n$, where  $v_j=(v_j^z,v_j^x)^\top$. 

To better understand the properties of the HW group, we need to introduce a symplectic structure on $\bbF_d^{2n}$. The symplectic product of $\bfu$ and $\bfv$, denoted by $[\bfu,\bfv]$, is defined as follows,
\begin{equation}\label{symprod}
	[\bfu,\bfv]:=\bfu^\top \Omega \bfv, \quad \Omega = \begin{pmatrix}
		0_{n} & \mathds{1}_{n} \\
		-\mathds{1}_{n} & 0_{n}
	\end{pmatrix}.
\end{equation}
The composition and commutation relations of   the two Weyl operators  $W_\bfu$ and $W_\bfv$
are determined by this symplectic product, 
\begin{align}
	W_\bfu W_\bfv=\omega_d^{[\bfu,\bfv]/2} W_{\bfu+\bfv},\quad  W_\bfu W_\bfv= \omega_d^{[\bfu,\bfv]} W_\bfv W_\bfu,\quad \bfu,\bfv\in \bbF_d^{2n}.
\end{align}
In addition, this symplectic product defines the symplectic group $\Sp(2n,d)$, which is composed of all $2n\times 2n$ matrices over $\bbF_d$ that preserve the symplectic product, that is,
\begin{equation}
	\Sp(2n,d):=\left\{M\in \bbF_d^{2n\times 2n} \ | \ M^\top \Omega M=\Omega\right\}. 
\end{equation}
Elements in $\Sp(2n,d)$ are called symplectic matrices or symplectic transformations. If $M$ is a symplectic matrix, then  $[M\bfu,M\bfv]=[\bfu,\bfv]$ for all  $\bfu,\bfv\in \bbF_d^{2n}$.

To understand the structure of the symplectic group $\Sp(2n,d)$. It is instructive to consider  the following nested subgroup chain of symplectic groups,
\begin{equation}
	G_1 \subset G_2 \subset \cdots \subset G_{n-1} \subset G_n=\Sp(2n,d),
\end{equation}
where $G_m$  is a shorthand for $\Sp(2m,d)$ for $m=1,2,\ldots, n$, and the embedding $G_{m-1} \mapsto G_m$ is realized as follows,
\begin{equation}
	M_{m-1} \mapsto \begin{pmatrix}
		1 & 0\\
		0 & 1
	\end{pmatrix} \oplus M_{m-1},\quad m=2,3,\ldots,n,\quad  M_{m-1} \in G_{m-1}. 
\end{equation} 
Moreover,  the map 
\begin{equation}\label{eq:subgroup}
	\begin{aligned}
		G_n/G_{n-1} \times G_{n-1}/G_{n-2} \times \cdots \times G_2/G_1 \times G_1 & \to \Sp(2n,d) \\
		\left([M_n],[M_{n-1}],\dots,[M_2],M_1\right) & \to M_n M_{n-1} \cdots M_2M_1
	\end{aligned}
\end{equation}
is a bijection. Notably, each symplectic transformation $M \in \Sp(2n,d)$ has a unique representation of the form $M_n M_{n-1} \cdots M_1$ with $[M_m]\in G_m/G_{m-1}$ for $m=2,3,\dots ,n$ and $M_1 \in G_1$. This fact would come useful when sampling a random symplectic matrix or constructing a desired symplectic transformation.

The Clifford group $\Cl(n,d)$ is the normalizer of the HW group $\caW(n,d)$, that is,
\begin{equation}
	\Cl(n, d) := \bigl\{ U \in \rmU\bigl(\caH_d^{\otimes n}\bigr) \,|\, U\caW(n,d)U^\dag = \caW(n,d)\bigr\},
\end{equation}
where $\rmU(\caH_d^{\otimes n})$ denotes the group of unitary operators on $\caH_d^{\otimes n}$. Up to overall phase factors, the Clifford group $\Cl(n,d)$ can be generated by the following single- and two-qudit gates (see Sec.~4.1.3 in \rcite{Hein21}):
\begin{equation}\label{eq:Clgen}
	\begin{aligned}
		F  :=\:& \frac{1}{\sqrt{d}} \sum_{a,b \in \bbF_d} \chi(ab) |a\>\<b|,  & U(\nu)  :=\:& \sum_{a\in\bbF_d} |\nu a\>\<a|, \\
		S(\nu)  :=\:&\sum_{a\in \bbF_d} \chi\bigl(2^{-1}\nu a^2\bigr)|a\>\<a|, &
		CX :=\:& \sum_{(a,b)\in \bbF_d^2} |a,a+b\>\<a,b|, \\
		Z  =\:&\sum_{a\in \bbF_d} \chi(a) |a\>\<a|, &
		X  =\:& \sum_{a\in\bbF_d} |a+1\>\<a|,
	\end{aligned}
\end{equation}
where $\nu\in \bbF_d^\times=\bbF_d\setminus \{0\}$ is a primitive element.
Since the commutation relations of Weyl operators are invariant under Clifford transformations by definition, every Clifford unitary induces a symplectic transformation on the  space $\bbF_d^{2n}$. Conversely,  given any symplectic transformation $M \in \Sp(2n,d)$, there exists a unitary operator $\mu(M)$ such that \cite{Gros06} 
\begin{equation}
	\mu(M)W_\bfu \mu(M)^\dag = W_{M\bfu} \quad \forall \bfu\in \bbF_d^{2n},
\end{equation}
where $\mu$ is called the Weil or metaplectic representation of the symplectic group \cite{Foll89,Weil64}. Note that  this conclusion relies on the assumption that $d$ is an odd prime and  is not guaranteed when $d=2$. Up to an overall  phase factor, any Clifford unitary $C \in \Cl(n,d)$ has the form
\begin{equation}\label{eq:CaM}
	C = W_\bfg \mu(M), \quad \bfg\in \bbF_d^{2n}, \quad M\in \Sp(2n,d).  
\end{equation}
In addition, we have
\begin{equation}\label{eq:CWC}
	C W_\bfu C^\dag = W_\bfg \mu(M) W_\bfu \mu(M)^\dag W_\bfg^\dag
	= W_\bfg  W_{M\bfu}  W_\bfg^\dag 
	= \chi([\bfg,M\bfu]) W_{M\bfu}.
\end{equation}

\subsection{Stabilizer codes and stabilizer states}
A subgroup $\caS$ of the HW group $\caW(n, d)$ is a \emph{stabilizer group} if it  does not contain the operator $\omega_d\bbI$, in which case $\caS$ is automatically an elementary Abelian group and $|\caS|\leq d^n$. Then the common eigenspace of $\caS$ with eigenvalue 1 defines a stabilizer code, denoted by $\caC_\caS$ henceforth. The projector onto 	$\caC_{\caS}$, known as a \emph{stabilizer projector}, can be expressed as 
\begin{equation}
	P_{\caS} = \frac{1}{|\caS|}\sum_{U\in\caS}U.
\end{equation}
The dimension of the stabilizer code $\caC_\caS$ reads 
\begin{align}
	\dim \caC_{\caS}=\rank P_{\caS}=\tr (P_{\caS})=\frac{d^n}{|\caS|}. 
\end{align}
If $\caS$ can be generated by $m$ Weyl operators, but cannot be generated by $m-1$ Weyl operators, then $|\caS|=d^m$ and $\dim \caC_{\caS}=d^{n-m}$. 

A stabilizer group $\caS$ in $\caW(n, d)$ is maximal if it has the maximum order of $d^n$. In this case, the stabilizer code $\caC_{\caS}$ has dimension 1 and is determined by  a normalized state, called  a \emph{stabilizer state}. A stabilizer state 
is uniquely determined by its stabilizer group, which is in turn determined by a  generating set composed of $n$ Weyl operators. 
The set of stabilizer states is denoted by $\Stab(n,d)$, which is also regarded as a normalized ensemble.
All stabilizer states in $\Stab(n,d)$ form one orbit under the action of the Clifford group $\Cl(n,d)$.

\section{Proofs of \pref{pro:LocalCl} and \thref{thm:LocalCl}}\label{sup:LocalCl}
In this section, we prove  \pref{pro:LocalCl} and \thref{thm:LocalCl}, which clarify the shadow norms associated with local Clifford measurements. To this end, we need to introduce some auxiliary notation and results by way of preparation.  

\subsection{Auxiliary results}
Let $\bfu,\bfv$ be two vectors in $\bbF_d^{2n}$. 
The weight $|\bfu|$ of $\bfu$ is defined as $|\bfu|:=\left|\{j\,|\, \bigl(u^z_j,u^x_j\bigr) \neq (0,0) \}\right|$. As a generalization, we define
\begin{equation}
	\begin{aligned}
		|\bfu \vee \bfv|& := \bigl|\bigl\{j\,|\, \bigl(u^z_j,u^x_j\bigr) \neq (0,0) \text{ or } \bigl(v^z_j,v^x_j\bigr) \neq (0,0)\bigr\}\bigr|, \\
		|\bfu \wedge \bfv|& := \bigl|\bigl\{j\,|\, \bigl(u^z_j,u^x_j\bigr) \neq (0,0) \text{ and } \bigl(v^z_j,v^x_j\bigr) \neq (0,0)\bigr\}\bigr|.
	\end{aligned}
\end{equation}
Evidently, we have
\begin{equation}
	|\bfu \vee \bfv|+|\bfu \wedge \bfv|=|\bfu|+|\bfv|.
\end{equation}
Next, we introduce two relations on $\bbF_d^{2n}$.
We write $\bfu  \triangleright \bfv$ if $u_j$ is proportional to (a multiple of) $v_j$ for $j=1,2,\ldots,n$ and write $\bfu \bowtie \bfv$  if $u_j, v_j$ are linearly dependent for $j=1,2,\ldots,n$.  In the latter case,  we can deduce the following relations,
\begin{align}
	|\bfu \vee \bfv|=\min \bigl\{|\bfs| :  \bfs \in \bbF_d^{2n}, \bfu \triangleright \bfs, \bfv \triangleright \bfs  \bigr\},\quad   |\bfu \wedge \bfv|=\max \bigl\{|\bfs| :  \bfs \in \bbF_d^{2n}, \bfs \triangleright \bfu, \bfs \triangleright \bfv  \bigr\}.
\end{align}
If $\bfu  \triangleright \bfv$, then $\bfu \bowtie \bfv$, but the converse is not guaranteed in general. Given $a,b\in \bbF_d$,  the Weyl operator  $\omega_d^a W_\bfu$ is $m$-local iff $|\bfu|=m$; the two Weyl operators  $\omega_d^a W_\bfu$ and $\omega_d^b W_\bfv$ are locally commutative iff $\bfu \bowtie \bfv$. 
\begin{lemma}\label{lem:WuWv}
	Suppose $d$ is an odd prime, $\bfb\in \bbF_d^n$,  and $\bfu,\bfv\in \bbF_d^{2n}$.  Then we have
	\begin{equation}\label{eq:WuWv1}
		\bbE_{U\sim \Cl(d)^{\otimes n}} \bigl(U^\dag|\bfb\> \<\bfb|U \<\bfb|U W_\bfu U^\dag|\bfb\> \<\bfb|U W_\bfv U^\dag|\bfb\>^*\bigr) = 
		\begin{cases}
			d^{-n}(d+1)^{-|\bfu\vee\bfv|} W_{\bfu-\bfv} \quad & \text{if  } \bfu \bowtie \bfv,\\
			0 \quad & \text{otherwise}.
		\end{cases}
	\end{equation}
\end{lemma}

\begin{proof}
	Thanks to the tensor structure of $U$, $W_\bfu$, $W_\bfv$, and $|\bfb\>$, it suffices to prove	\eref{eq:WuWv1}	in  the case $n = 1$. Then we can rewrite the LHS of \eref{eq:WuWv1} as follows,
	\begin{equation}\label{eq:WuWv2}
		\bbE_{|\Psi\>\sim\Stab(d)}\bigl( |\Psi\>\<\Psi| \<\Psi| W_\bfu |\Psi\> \<\Psi| W_\bfv |\Psi\>^*\bigr)
		= \frac{1}{d(d+1)} \sum_{|\Psi\> \in \Stab(d)} \bigl(|\Psi\>\<\Psi| \<\Psi| W_\bfu |\Psi\> \<\Psi| W_\bfv |\Psi\>^*\bigr),
	\end{equation}	
	where $\Stab(d)=\Stab(1,d)$ is the set of single-qudit stabilizer states. 
	
	If the condition $\bfu \bowtie \bfv$ does not hold, then $W_\bfu$ and $W_\bfv$ do not commute (given that $n=1$ by assumption) and cannot belong to a same stabilizer group. Therefore, either $\<\Psi| W_\bfu |\Psi\>=0$ or $\<\Psi| W_\bfv |\Psi\>=0$, that is, $\<\Psi| W_\bfu |\Psi\>\<\Psi| W_\bfv |\Psi\>^*=0$, which implies \eref{eq:WuWv1}.

	In the rest of this proof we assume that $\bfu \bowtie \bfv$, which means  $W_\bfu$ and $W_\bfv$ commute with each other. If $W_\bfu=W_\bfv=\bbI$, then \eref{eq:WuWv1} holds because $\bbE_{|\Psi\>\sim\Stab(d)} |\Psi\>\<\Psi|=\bbI/d$. 
	
	Next, suppose $W_\bfu$ and $W_\bfv$ commute, and at least one of them is not proportional to the identity operator. Then we can find a Clifford unitary $C\in \Cl(d)$ such that
	\begin{equation}
		C^{\dag}W_\bfu C = \omega_d^{r_1}Z^{s_1}, \quad C^{\dag}W_\bfv C = \omega_d^{r_2}Z^{s_2},
	\end{equation}
	where $r_1,s_1,r_2,s_2 \in \bbF_d$. Therefore,
	\begin{align}
		&\sum_{|\Psi\> \in \Stab(d)} \bigl(|\Psi\>\<\Psi| \<\Psi| W_\bfu |\Psi\> \<\Psi| W_\bfv |\Psi\>^*\bigr)=
		\sum_{|\Psi\> \in \Stab(d)}\bigl( |\Psi\>\<\Psi| \<\Psi| C\omega_d^{r_1}Z^{s_1}C^\dag |\Psi\> \<\Psi|  C\omega_d^{r_2}Z^{s_2}C^\dag |\Psi\>^*\bigr)\nonumber\\
		&=\sum_{|\Psi\> \in \Stab(d)} \bigl(C|\Psi\>\<\Psi| C^\dag \<\Psi| \omega_d^{r_1}Z^{s_1} |\Psi\> \<\Psi| \omega_d^{r_2}Z^{s_2}|\Psi\>^*\bigr)=\sum_{b\in \bbF_d} \bigl(C|b\>\<b| C^\dag \<b| \omega_d^{r_1}Z^{s_1} |b\> \<b| \omega_d^{r_2}Z^{s_2}|b\>^*\bigr)\nonumber\\
		&=C\omega_d^{r_1}Z^{s_1}\omega_d^{-r_2}Z^{-s_2} C^\dag=W_{\bfu-\bfv}. 
	\end{align}	
	Together with \eref{eq:WuWv2}, this equation implies \eref{eq:WuWv1} and completes the proof of \lref{lem:WuWv}. 
\end{proof}

For local Clifford measurements, the reconstruction map $\caM^{-1}$ in \eref{eq:reconlocal}  can be expressed as
\begin{equation}\label{eq:LocalReconstruct}
	\caM^{-1}= \left(\caD_{1/(d+1)}^{-1}\right)^{\otimes n},
\end{equation}
where $\caD_y^{-1}(\cdot)$ is the inverse of the single-qudit depolarizing channel $\caD_y(\cdot)=y\cdot+ (1-y)\frac{\tr(\cdot)}{d}I$. Note that $\caD_{1/(d+1)}^{-1}(\Ob)=(d+1)\Ob-\tr (\Ob) I$ for any linear operator $\Ob$ on $\caH_d$.  If $\Ob=W_u$ is a single-qudit Weyl operator with $u\in \bbF_d^2$, then 
\begin{equation}\label{eq:LocalReconstructQudit1}
	\caD^{-1}_{\sfrac{1}{(d\!+\!1)}}(W_{u}) =
	\begin{cases}
		W_u &\quad \text{if  } W_u \propto I,\\
		(d+1)W_u &\quad \text{otherwise}.
	\end{cases}
\end{equation}

\subsection{Proof of \pref{pro:LocalCl}}
\begin{proof}[Proof of \pref{pro:LocalCl}]	
	If $d=2$, then \pref{pro:LocalCl} holds according to Lemma~3 of \rcite{HuangKP20}.	
	
	Next, suppose $d$ is an odd prime and $\Ob$ is an $m$-local Weyl operator. Then $\Ob$ has the form $\Ob=\rme^{\rmi \phi}W_\bfu$, where $\phi$ is a real phase and $\bfu\in \bbF_d^{2n}$ with $|\bfu|=m$.	By virtue of \eqsref{eq:sndef}{eq:LocalReconstruct} we can deduce that
	\begin{align}\label{eq:recon weyl}
		\|\Ob\|_\sh^2 &=\|W_\bfu\|_\sh^2  = \max_{\sigma} \sum_{\bfb\in \bbF_d^n}\bbE_{U \sim \Cl(d)^{\otimes n}} \Biggl[\<\bfb|U\sigma U^{\dag}|\bfb\> \biggl|\<\bfb|U\Bigl(\caD^{-1}_{1/(d+1)}\Bigr)^{\otimes n}(W_\bfu)U^{\dag} |\bfb\>\biggr|^2\Biggr] \nonumber\\	
		& = (d+1)^{2m} \max_{\sigma} \sum_{\bfb\in \bbF_d^n}\bbE_{U \sim \Cl(d)^{\otimes n}}\Bigl[ \<\bfb|U\sigma U^{\dag}|\bfb\> \bigl|\<\bfb|UW_\bfu U^{\dag} |\bfb\>\bigr|^2\Bigr] \nonumber\\
		& = d^n(d+1)^{2m} \times \frac{1}{d^n(d+1)^m} \max_{\sigma} \tr\left(\sigma W_{\bfu-\bfu} \right) = (d+1)^m,
	\end{align}
	which confirms \pref{pro:LocalCl}.
	Here the third equality follows from \eref{eq:LocalReconstructQudit1} and the fact that $W_\bfu$ is $m$-local,  the fourth equality follows from \lref{lem:WuWv}, and the last equality holds because $W_{\bfu-\bfu}=\bbI$. 
\end{proof}

\subsection{Proof of \thref{thm:LocalCl}}\label{sup:SnLocal}
\begin{proof}[Proof of \thref{thm:LocalCl}]
	If  $d=2$, then  \eref{eq:local linear} follows from  Proposition 3 of \rcite{HuangKP20}, so it remains to  consider the case in which $d$ is an odd prime.
	
	We expand $\tOb$ in the Weyl basis as $\tOb = \sum_{\bfu \in \bbF_d^{2m}} \alpha_\bfu W_\bfu$. 
	According to \eref{eq:LocalReconstructQudit1}, we have
	\begin{equation}\label{eq:recon_m}
		\Bigl(\caD_{\sfrac{1}{(d\!\!+\!\!1)}}^{-1}\Bigr)^{\otimes m} (\tOb) = \sum_{\bfu \in \bbF_d^{2m}} (d+1)^{|\bfu|} \alpha_\bfu W_\bfu.
	\end{equation}
	Define $\scrV_1^*\subset \bbF_d^2$ and $\scrV_m^*\subset \bbF_d^{2m}$ as follows,
	\begin{equation}
		\begin{aligned}
			\scrV_1^*&:=\bigl\{(0,1)^\top,(1,1)^\top,\ldots,(d-1,1)^\top, (1,0)^\top\bigr\},\\
			\scrV_m^* &:= \bigl\{(s^z_1,s^z_2,\ldots,s^z_m,s^x_1,s^x_2,\ldots,s^x_m)^\top\,|\, (s^z_j,s^x_j)^\top \in \scrV_1^*  \  \forall 1\leq j \leq m\bigr\}.
		\end{aligned}	
	\end{equation}
	In conjunction with  \eqsref{eq:sndef}{eq:LocalReconstruct}  we can deduce that
	\begin{align}
		\|\Ob\|_\sh^2 & =	\| \tOb \|_\sh^2= \max_{\sigma} \bbE_{U \sim \Cl(d)^{\otimes m}} \sum_{\bfb\in \bbF_d^m} \Bigl[ \<\bfb|U\sigma U^{\dag}|\bfb\> \Bigl|\<\bfb|U\Bigl(\caD_{\sfrac{1}{(d\!\!+\!\!1)}}^{-1}\Bigr)^{\otimes m} (\tOb)U^{\dag} |\bfb\>\Bigr|^2 \biggr]\nonumber\\
		& = \max_{\sigma} \sum_{\bfb\in \bbF_d^m} \sum_{\bfu,\bfv \in \bbF_d^{2m}} \Bigl\{(d+1)^{|\bfu|+|\bfv|} \alpha_\bfu \alpha_\bfv^* \tr\bigl[\sigma \bbE_{U \sim \Cl(d)^{\otimes m}}  \bigl(U^\dag|\bfb\>\<\bfb|U \<\bfb|U W_\bfu U^\dag|\bfb\> \<\bfb|U W_\bfv U^\dag|\bfb\>^*\bigr) \bigr] \Bigr\} \nonumber \\	
		& = \max_{\sigma}  \sum_{\bfu,\bfv\in  \bbF_d^{2m}, \bfu\bowtie \bfv} \biggl[\frac{(d+1)^{|\bfu|+|\bfv|} }{(d+1)^{|\bfu\vee\bfv|}}\alpha_\bfu \alpha_\bfv^* \tr(\sigma W_{\bfu-\bfv})\biggr] =\left\|\sum_{\bfu,\bfv\in  \bbF_d^{2m}, \bfu\bowtie \bfv} \frac{(d+1)^{|\bfu|+|\bfv|} }{(d+1)^{|\bfu\vee\bfv|}}\alpha_\bfu \alpha_\bfv^*  W_{\bfu-\bfv}\right\| \nonumber \\	
		& =\left\|\sum_{\bfs \in \scrV_m^*} \sum_{\bfu,\bfv \triangleright \bfs}\frac{(d+1)^{|\bfu|+|\bfv|}}{(d+1)^m} \alpha_\bfu \alpha_\bfv^* W_{\bfu-\bfv}\right\| \leq \frac{1}{(d+1)^m} \sum_{\bfs \in \scrV_m^*} \left\|\sum_{\bfu,\bfv \triangleright \bfs} (d+1)^{|\bfu|+|\bfv|} \alpha_\bfu \alpha_\bfv^*  W_{\bfu-\bfv} \right\| \nonumber \\
		& = \frac{1}{(d+1)^m} \sum_{\bfs \in \scrV_m^*} \left\|\sum_{\bfu \triangleright \bfs} (d+1)^{|\bfu|} \alpha_\bfu  W_\bfu \right\|^2\leq \frac{1}{(d+1)^m} \sum_{\bfs \in \scrV_m^*} \left|\sum_{\bfu \triangleright \bfs} (d+1)^{|\bfu|} \alpha_\bfu \right|^2, \label{eq:LocalCliffordProof}
	\end{align}
	where each maximization is over all quantum states on $\caH_d^{\otimes m}$. Here the third equality follows from \eqref{eq:recon_m}, the fourth equality follows from \lref{lem:WuWv},
	and the sixth equality holds because 
	\begin{align}
		|\{\bfs\in \scrV_m^*\,| \,\bfu,\bfv \triangleright \bfs\}|=(d+1)^{m-|\bfu\vee\bfv|}. 
	\end{align}

	Next, by virtue of  the Cauchy-Schwartz inequality applied  to \eref{eq:LocalCliffordProof} we can deduce that
	\begin{align}
		\| \Ob\|_\sh^2 &\leq 	\frac{1}{(d+1)^m} \sum_{\bfs \in \scrV_m^*} \left|\sum_{\bfu \triangleright \bfs} (d+1)^{|\bfu|} \alpha_\bfu \right|^2  \leq \frac{1}{(d+1)^m} \sum_{\bfs \in \scrV_m^*} \left( \sum_{\bfu \triangleright \bfs} (d+1)^{|\bfu|}\right) \left( \sum_{\bfu \triangleright \bfs}(d+1)^{|\bfu|}|\alpha_\bfu|^2 \right)\nonumber\\
		& = d^{2m}\sum_{\bfs \in \scrV_m^*} \sum_{\bfu \triangleright \bfs} (d+1)^{|\bfu|-m}|\alpha_\bfu|^2  = d^{2m} \sum_{\bfu \in \bbF_d^{2m}} |\alpha_\bfu|^2 = d^m \|\tOb\|_2^2,
	\end{align} 
	which confirms  \thref{thm:LocalCl}. In deriving the above equalities we have taken into account the following facts,
	\begin{equation}
		\sum_{\bfu \triangleright \bfs} (d+1)^{|\bfu|} = \sum_{j=0}^m 
		\binom{m}{j}
		(d-1)^j(d+1)^j = d^{2m},\quad |\{\bfs\in \scrV_m^*\,|\, \bfu \triangleright \bfs \}|=(d+1)^{m-|\bfu|}. 
	\end{equation}
\end{proof}

To understand the connection between \pref{pro:LocalCl} and \thref{thm:LocalCl}, let
$\tOb=W_\bfq$ with $\bfq\in \bbF_d^{2m}$ be an $m$-local Weyl operator. Then $\Ob$ is also an $m$-local Weyl operator, and we have $\alpha_\bfu=1$ if $\bfu=\bfq$ and $\alpha_\bfu=0$ otherwise. So  both inequalities in \eref{eq:LocalCliffordProof}  are saturated and we recover \pref{pro:LocalCl}.

\section{Proof sketches of \thsref{thm:GlobalCl}-\ref{thm:Clifford+T}}\label{sup:GlobalClProofs}
In this section we sketch the proofs of \thsref{thm:GlobalCl}-\ref{thm:Clifford+T} based on our companion paper \cite{ZhuMY24}. 

Suppose the unitary ensemble $\caE$ used in shadow estimation forms a unitary 2-design. Denote by  $Q(\caE)$ and $\bQ(\caE)$  the third moment operator and normalized moment operator associated with the corresponding measurement ensemble, that is,
\begin{align}
	Q(\caE)&:=\frac{1}{D} \bbE_{U\sim \caE}\sum_{\bfb\in \bbF_d^n} \bigl(U^\dag |\bfb\>\<\bfb|U\bigr)^{\otimes 3},  \label{eq:QE3}\\
	\bQ(\caE)&:=\frac{D(D+1)(D+2)}{6}Q(\caE)=\frac{(D+1)(D+2)}{6} \bbE_{U\sim \caE}\sum_{\bfb\in \bbF_d^n} \bigl(U^\dag |\bfb\>\<\bfb|U\bigr)^{\otimes 3},
\end{align} 
which are Hermitian operators on $\caH^{\otimes 3}$ with  $\caH=\caH_d^{\otimes n}$. For the convenience of  discussion, we  label the three copies of $\caH$ by $A, B, C$, respectively. 
Then the shadow norm in \eref{eq:sndef} can be expressed as follows,
\begin{align}
	\| \Ob \|_\sh^2&=\max_{\sigma}  \bbE_{U\sim \caE} \sum_{\bfb\in \bbF_d^n} \Bigl[ \<\bfb|U\sigma U^\dag |\bfb\> \left|\<\bfb|U\caM^{-1}(\Ob)U^\dag |\bfb\>\right|^2\Bigr]
	=\max_{\sigma} D\tr\bigl\{Q(\caE)\bigl[\sigma\otimes\caM^{-1}(\Ob)\otimes\caM^{-1}(\Ob)^\dag\bigr]\bigr\}
	\nonumber\\
	&=\max_{\sigma} D(D+1)^2\tr\bigl[Q(\caE)\bigl(\sigma\otimes\Ob\otimes\Ob^{\dag}\bigr)\bigl]=\max_{\sigma} \frac{6(D+1)}{D+2}\tr\bigl[\bQ(\caE)\bigl(\sigma\otimes\Ob\otimes\Ob^{\dag}\bigr)\bigr]\nonumber\\
	&=\frac{6(D+1)}{D+2}\bigl\|\tr_{BC}\bigl[\bQ(\caE)\bigl(\bbI\otimes\Ob\otimes\Ob^{\dag}\bigr)\bigr]\bigr\|,
\end{align}
where the third equality holds because $\caM^{-1}(\Ob)=(D+1)\Ob$ given that $\caE$ is a 2-design and $\Ob$ is traceless. The above expression for the shadow norm has the same form as in Eq.~(14) in our companion paper \cite{ZhuMY24}. 

Next, we turn to shadow estimation based on the Clifford group supplemented by $T$ gates. 	Now the moment operator  $Q(\caE)$ is identical to the moment operator of a suitable Clifford orbit, as we shall see shortly. Recall that each $T$ gate on $\caH_d$  is determined by a cubic function according to Appendix B and the companion paper~\cite{ZhuMY24}. In addition, 
for each $T$ gate we can define a magic state as follows:
\begin{align}\label{eq:TMagic}
	|T\>:=TF|0\>=T|+\>,
\end{align}
where $|+\>:=F|0\>=\sum_{b\in \bbF_d}|b\>/\sqrt{d}$.  By definition, $|T\>$ can be generated from $|0\>$ by the magic gate $T$ following the Fourier gate $F$.
In addition, the conjugate $T^\dag$ is also a diagonal magic gate and the corresponding magic state reads
\begin{align}\label{eq:TdagMagic}
	|T^\dag\>=T^\dag F|0\>=T^\dag F^\dag|0\>,
\end{align}
where the second inequality holds because $ F|0\>= F^\dag|0\>=|+\>$. Incidentally, $T^\dag$ and $T$ have the same cubic character given that $\eta_3(-1)=1$ for any cubic character $\eta_3$ on $\bbF_d^{\times}$ (see  Appendix B and \rcite{ZhuMY24} for additional information).

\begin{lemma}\label{lem:MomentCl}
	Suppose  $\caE=\Cl(n,d)$; then 
	\begin{align}
		Q(\caE)=\bbE_{U\sim \caE}\bigl(U^\dag |\mathbf{0}\>\<\mathbf{0}|U\bigr)^{\otimes 3}, \label{eq:MomentCl}
	\end{align} 
	where $|\mathbf{0}\>$ is a shorthand for $|0\>^{\otimes n}$. 
\end{lemma}

\begin{lemma}\label{lem:MomentClT}
	Suppose  $\caE=\caE\left(\{T_j\}_{j=1}^k\right)$ and $|\Psi\>=\bigl|T_1^\dag\bigr\>\otimes \cdots\otimes \bigl|T_k^\dag\bigr\>\otimes |0\>^{\otimes (n-k)}$, where $\bigl|T_j^\dag\bigr\>$ for $j=1, 2, \ldots, k$   are defined according to \eref{eq:TdagMagic}.  Then 
	\begin{align}
		Q(\caE)=\bbE_{U\sim \Cl(n,d)}\bigl(U^\dag |\Psi\>\<\Psi|U\bigr)^{\otimes 3}. \label{eq:MomentClT}
	\end{align} 
\end{lemma}
\Lref{lem:MomentCl} shows that the moment operator $Q(\caE)$ associated with the Clifford ensemble is identical to the counterpart of the orbit $\Stab(n,d)$ of stabilizer states (and similarly for the normalized moment operator). This result is expected  given that all computational-basis states form one orbit under the action of the Clifford group. It highlights the importance of the moment operator in studying shadow estimation.
The key properties of the  moment operator of $\Stab(n,d)$ have been clarified in our companion paper \cite{ZhuMY24} after systematic and in-depth study. Thanks to \Lref{lem:MomentCl}, \thsref{thm:GlobalCl} and \ref{thm:GlobalStabProj} follow from Theorems~4 and 3 in the companion paper, respectively. 
\Lref{lem:MomentClT} further shows that the  moment operator $Q(\caE)$ associated with  the ensemble  $\caE=\caE\left(\{T_j\}_{j=1}^k\right)$ is identical to the moment operator of a Clifford orbit generated from a magic state. Therefore, \thref{thm:Clifford+T} follows from Theorems~15 and 21 in the companion paper \cite{ZhuMY24}.	In the following proof we use $X^\bfb$ as a shorthand for $X^{b_1}\otimes \cdots\otimes X^{b_n}$. 
\begin{proof}[Proof of \lref{lem:MomentCl}]By virtue of the definition in \eref{eq:QE3} we can deduce that
	\begin{align}
		Q(\caE)&=\frac{1}{D} \bbE_{U\sim \caE}\sum_{\bfb\in \bbF_d^n} \bigl(U^\dag |\bfb\>\<\bfb|U\bigr)^{\otimes 3}
		=\frac{1}{D} \bbE_{U\sim \caE}\sum_{\bfb\in \bbF_d^n} \bigl(U^\dag X^{\bfb} |\mathbf{0}\>\<\mathbf{0}| X^{-\bfb}U\bigr)^{\otimes 3}
		=\bbE_{U\sim \caE}\bigl(U^\dag |\mathbf{0}\>\<\mathbf{0}|U\bigr)^{\otimes 3},
	\end{align} 
	which confirms \eref{eq:MomentCl}. Here the second equality holds because $|\bfb\>=X^\bfb|\mathbf{0}\>$, and the third equality holds because $X^\bfb\in \Cl(n,d)$.	
\end{proof}

\begin{proof}[Proof of \lref{lem:MomentClT}]
	Let $V=FT_1\otimes \cdots\otimes FT_k\otimes I^{\otimes(n-k)}$. Then  
	\begin{align} \label{eq:MomentClTproof}
		|\Psi\>=V^\dag|\mathbf{0}\>,\quad V^\dag X^\bfb&= V^\dag X^\bfb V V^\dag= \bigl(Z^{-b_1}\otimes\cdots\otimes  Z^{-b_k}\otimes X^{b_{k+1}}\otimes\cdots\otimes X^{b_n}\bigr)V^\dag,  
	\end{align}
	where the third equality holds because $T_j^\dag F^\dag XF T_j=T_j^\dag Z^{-1} T_j=Z^{-1}$ for $j=1,2,\ldots, k$. In addition, $\caE=\caE\left(\{T_j\}_{j=1}^k\right)=V\Cl(n,d)$ and
	\begin{align}
		Q(\caE)&=\frac{1}{D} \bbE_{U\sim \caE}\sum_{\bfb\in \bbF_d^n} \bigl(U^\dag |\bfb\>\<\bfb|U\bigr)^{\otimes 3}
		=\frac{1}{D} \bbE_{U\sim \Cl(n,d)}\sum_{\bfb\in \bbF_d^n} \bigl(U^\dag V^\dag X^\bfb|\mathbf{0}\>\<\mathbf{0}|X^{-\bfb}V  U\bigr)^{\otimes 3}\nonumber\\
		&=\frac{1}{D} \bbE_{U\sim \Cl(n,d)} \sum_{\bfb\in \bbF_d^n}\bigl(U_\bfb^\dag |\Psi\>\<\Psi|  U_\bfb\bigr)^{\otimes 3}= \bbE_{U\sim \Cl(n,d)} \bigl(U^\dag |\Psi\>\<\Psi|  U\bigr)^{\otimes 3},
	\end{align}
	where $U_\bfb=\bigl(Z^{b_1}\otimes\cdots\otimes  Z^{b_k}\otimes X^{-b_{k+1}}\otimes\cdots\otimes X^{-b_n}\bigr)U$ and the third equality follows from \eref{eq:MomentClTproof}. This equation confirms   \eref{eq:MomentClT} and completes the proof of \lref{lem:MomentClT}.
\end{proof}

\section{Additional numerical results}\label{sup:Numerical}

\begin{figure}[bp]
	\centering 
	\includegraphics[width=0.7\textwidth]{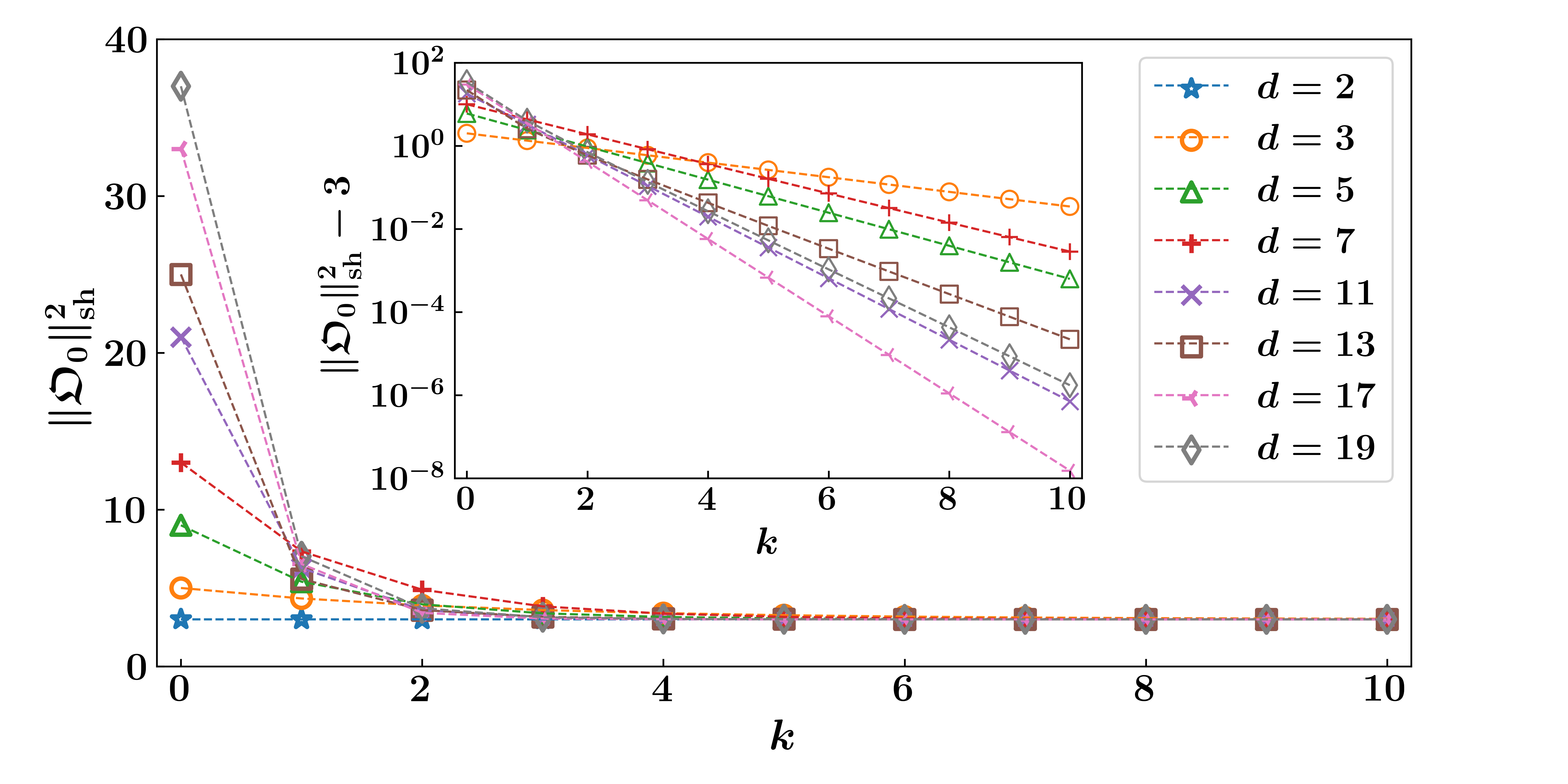}
	\caption{Shadow norm associated with the projector $\Ob$ onto an $n$-qudit  stabilizer state with respect to the Clifford circuit supplemented by $k$ canonical $T$ gates. Here $n=50$ and $\Ob_0$ denotes the traceless part of $\Ob$.}
	\label{fig:stabstatenorm2}
\end{figure}

\begin{figure}[tbp]
	\centering 
	\includegraphics[width=0.85\textwidth]{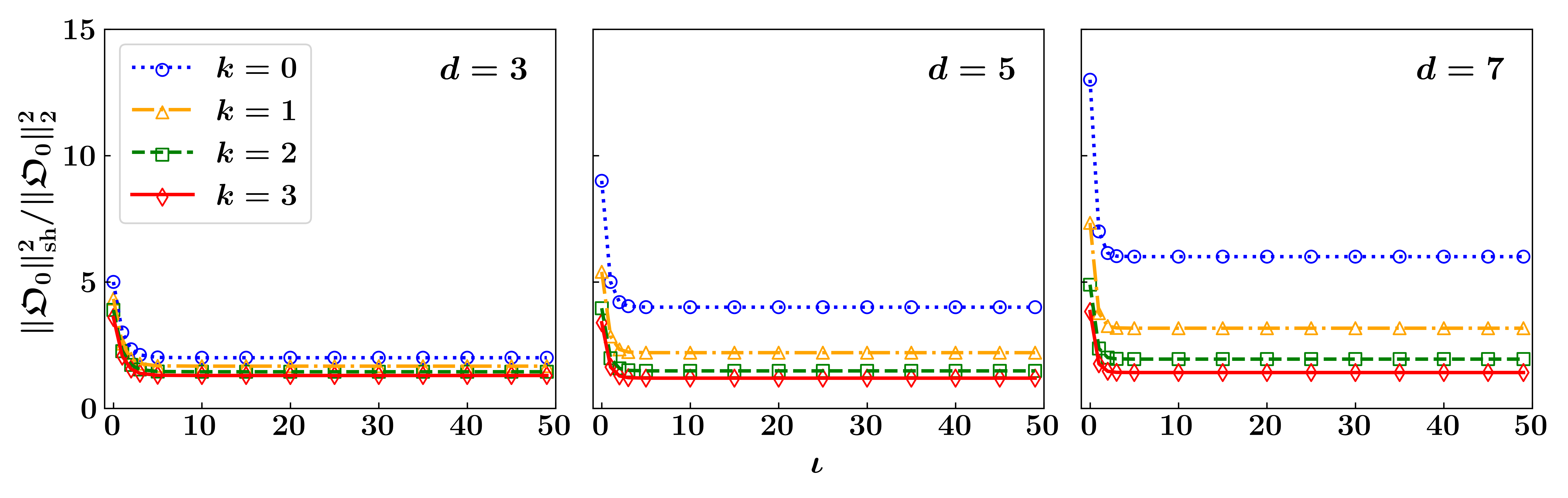}
	\caption{The ratio between the shadow norm $\|\Ob_0\|_\sh^2$ and $\|\Ob_0\|_2^2$. Here $\Ob$ is an $n$-qudit rank-$K$ stabilizer projector with $n=50$,  $K=d^\iota$, and $\|\Ob_0\|_2^2=K-(K^2/D)$, where $\iota$ is an integer and $0 \leq \iota \leq n-1$. The unitary ensemble is based on the Clifford circuit supplemented by $k$ canonical $T$ gates. 
	}
	\label{fig:stabprojnorm}
\end{figure}

\subsection{Shadow norms of stabilizer states and stabilizer projectors}
In the main text we have shown that, compared with the qubit case, the overhead of qudit shadow estimation based on the Clifford group is almost independent of the qudit number $n$.
Moreover, this overhead vanishes exponentially with the addition of $T$ gates; actually, a single $T$ gate can already suppress the overhead to a constant factor that is independent of the local dimension, as illustrated  in \fref{fig:stabstatenorm1}. Here we provide some additional results on the shadow norms associated with stabilizer states and stabilizer projectors. 
\Fref{fig:stabstatenorm2} shows the shadow norm associated with the projector $\Ob$ onto an $n$-qudit stabilizer state as a function of the local dimension $d$ and the number $k$ of canonical $T$ gates, where $\Ob_0$ denotes the traceless part of $\Ob$. The deviation $\|\Ob_0\|_\sh^2-3$ decreases exponentially with $k$, and it tends to decrease more quickly as the local dimension increases, which agrees with \thref{thm:Clifford+T}. 
Note that   $\|\Ob_0\|_\sh^2 \approx 3$ if the measurement primitive is a 3-design, which is the case in the qubit setting. Therefore, the Clifford circuit supplemented by a few $T$ gates can achieve almost the same performance as 3-design ensembles.
\Fref{fig:stabprojnorm} further shows the shadow norms associated with normalized stabilizer projectors.

\subsection{\label{sup:TgateImpacts}Impacts of different $T$ gates on shadow estimation}
When $d=1\mmod 3$, we can distinguish three types of $T$ gates depending on the cubic characters of the cubic coefficients of the underlying cubic polynomials as mentioned in Appendix B and discussed in more detail in our companion paper \cite{ZhuMY24}. Note that the cubic coefficients can be expressed as powers of a primitive element in $\bbF_d$. For reference, in \tref{tab:nu} we list a specific choice of a primitive element for $\bbF_d$ for each odd prime $d<50$ that we use in this paper and the companion paper \cite{ZhuMY24}.

\begin{table*}
	\renewcommand{\arraystretch}{1.8}
	\caption{\label{tab:nu} Specific choice of a primitive element for $\bbF_d$ with $d<50$.}	
	\centering
	\begin{tabular}{{p{1.2cm}<{\centering}|p{0.8cm}<{\centering}p{0.8cm}<{\centering}p{0.8cm}<{\centering}p{0.8cm}<{\centering}p{0.8cm}<{\centering}p{0.8cm}<{\centering}p{0.8cm}<{\centering}p{0.8cm}<{\centering}p{0.8cm}<{\centering}p{0.8cm}<{\centering}p{0.8cm}<{\centering}p{0.8cm}<{\centering}p{0.8cm}<{\centering}p{0.8cm}<{\centering}p{0.8cm}<{\centering}}}
		\hline\hline
		$d$ & $3$ & $5$ & $7$ & $11$ & $13$ & $17$ & $19$ & $23$ & $29$ & $31$ & $37$ & $41$ & $43$ & $47$\\
		\hline		
		$\nu$ & $2$ & $2$ & $3$ & $2$ & $2$ & $3$ & $2$ & $5$ & $2$ & $3$ & $2$ & $6$ & $3$ & $5$\\
		\hline\hline
	\end{tabular}
\end{table*}

\begin{table*}
	\renewcommand{\arraystretch}{1.8}
	\caption{\label{tab:Ttype} Optimal choices  of one and two $T$ gates (in minimizing the shadow norms) in fidelity estimation of $n$-qudit stabilizer states with $n=100$ and $d=7,13,19,31$. Here $\nu$ is a primitive element in the finite field $\bbF_d$ (see \tref{tab:nu}). Different  $T$ gates are distinguished by the cubic coefficients of the underlying cubic polynomials expressed as powers of $\nu$; exponents of $\nu$ are distinguished only modulo 3. }	
	\centering
	\begin{tabular}{{p{1.5cm}<{\centering}|p{3cm}<{\centering}p{3cm}<{\centering}p{3cm}<{\centering}p{3cm}<{\centering}}}
		\hline\hline
		$d$ & $\nu$ & $T$ & $2T$ \; (identical) & $2T$ \\
		\hline		
		$7$ & $3$ & $[\nu^2]$ & $[\nu^2, \nu^2]$ &  $[\nu, \nu^2]$ \\	
		$13$ & $2$ & $[1]$ & $[1, 1]$ &  $[1, \nu]$ \\
		$19$ & $2$ & $[\nu^2]$ & $[\nu^2, \nu^2]$ &  $[1, \nu^2]$ \\
		$31$ & $3$ & $[1]$ & $[1, 1]$ &   $[1, \nu^2]$ \\
		\hline\hline
	\end{tabular}
\end{table*}

Here we illustrate the impacts of different $T$ gates on the performance in shadow estimation. \Tref{tab:Ttype} lists the optimal choices of one and two $T$ gates in fidelity estimation of $n$-qudit stabilizer states with $n=100$ and  $d=7,13,19,31$. \Fref{fig:sn7} illustrates the significant impacts of different $T$ gates on the shadow norms of random two-qudit traceless (diagonal) observables and on the MSE in fidelity estimation of the 100-qudit GHZ state with $d=7$. The optimal choices of one and two $T$ gates are consistent with the results in \Tref{tab:Ttype}.

\begin{figure*}[tbp]
	\centering 
	\includegraphics[width=0.85\textwidth]{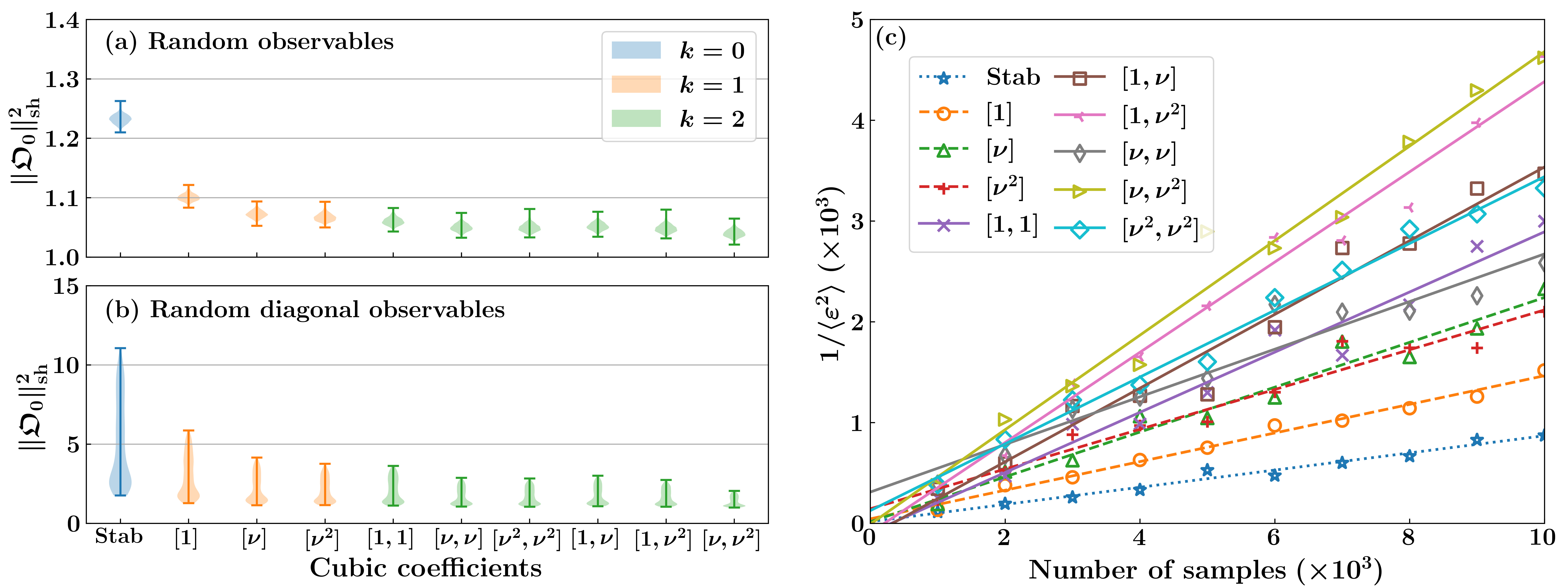}
	\caption{Impacts of different $T$ gates on qudit shadow estimation with $d=7$. The unitary ensemble $\caE$ is based on the Clifford circuit supplemented by $k$  qudit $T$ gates. The  case $k=0$ is labeled by `Stab' because  the corresponding measurement primitive is $\Stab(n,d)$; when $k=1,2$, 
		different $T$ gates are distinguished by the cubic coefficients expressed as powers of $\nu=3$ as in \tref{tab:Ttype}.  (a,b) Distributions of the shadow norms of random two-qudit traceless (diagonal)  observables  normalized with respect to the Hilbert-Schmidt norm. (c) Simulation results on the inverse MSE in estimating the fidelity of the $100$-qudit GHZ state. The MSE $\<\epsilon^2\>$ for each data point is the average over 100 runs. The dotted, dashed, and solid lines, determined by interpolation,
		correspond to the cases  $k=0,1,2$, respectively.}
	\label{fig:sn7}
\end{figure*}

\begin{figure*}[tbp]
	\centering 
	\includegraphics[width=0.8\textwidth]{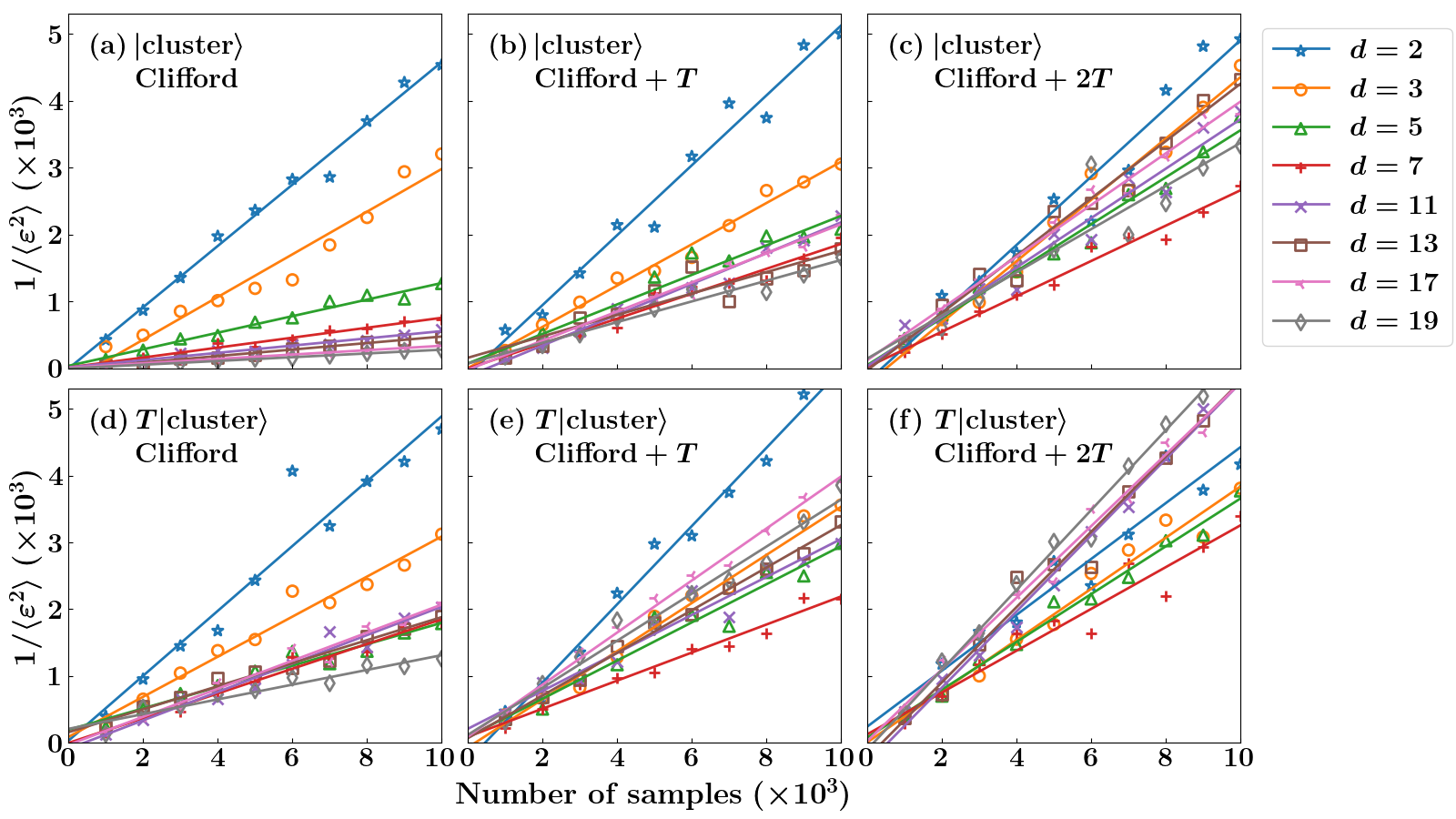}
	\caption{Simulation results on the inverse MSE in estimating the fidelity of the $100$-qudit cluster state $|\mathrm{cluster}\>$  defined on a $10\times 10$ square lattice with periodic boundary conditions. The estimation protocols are based on the Clifford circuit supplemented by up to two canonical $T$ gates. The MSE $\<\epsilon^2\>$ for each data point is the average over 100 runs, and the solid lines are determined by interpolation. Results on the state $T |\mathrm{cluster}\>$ are also shown for comparison.}
	\label{fig:cluster}
\end{figure*}

\subsection{Simulation results on cluster states}
Here we provide additional simulation results on fidelity estimation of cluster states based on the Clifford circuit supplemented by up to two canonical $T$ gates. \Fref{fig:cluster} shows the inverse MSE in estimating the fidelity of the 100-qudit cluster state  defined on a $10\times 10$ square lattice with periodic boundary conditions. The overall behavior is similar to the counterpart of the GHZ state as shown in \fref{fig:GHZ} in the main text. 

\begin{figure*}[tp]
	\centering 
	\includegraphics[width=0.8\textwidth]{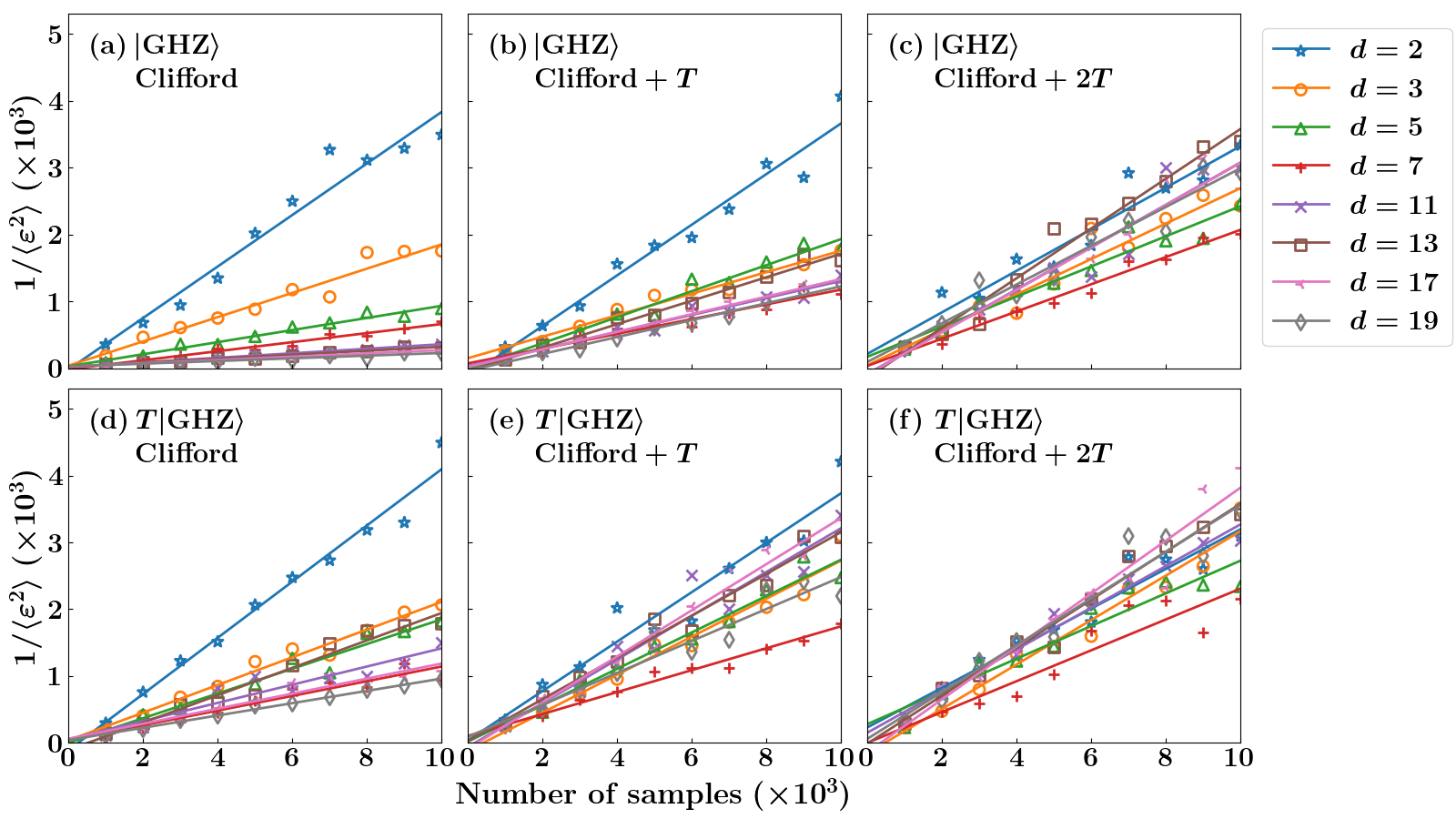}
	\caption{Simulation results on the inverse MSE in estimating the fidelity of the $100$-qudit GHZ state $|\GHZ\>$ using the median-of-means estimation with $L=10$. Here the raw data are the same as that employed in \fref{fig:GHZ} in the main text. 
		The MSE $\<\epsilon^2\>$ for each data point is the average over 100 runs, and the solid lines are determined by interpolation as before. Results on the state $T |\GHZ\>$ are also shown for comparison. }
	\label{fig:GHZ,mom}
\end{figure*}

\subsection{Median-of-means estimation}
So far the estimator employed in our numerical simulation is based on the empirical mean in \eref{eq:means}. As an alternative, we can also employ the median of means \cite{HuangKP20}. To this end, we first divide the $N$ samples into $L$ groups, each containing $N/L$ samples, assuming that $N$ is a multiple of $L$. For each group we can construct an  estimator based on the empirical mean. In this way, we can construct $L$ preliminary estimators. Now, the final estimator is the median of the $L$ preliminary estimators. Compared with the estimator in \eref{eq:means}, this alternative estimator can effectively reduce the probability of large deviation at the prize of a slightly larger MSE. 

\Fref{fig:GHZ,mom} illustrates the performance of this estimator with $L=10$ in estimating the fidelity of the $100$-qudit GHZ state based on the qudit Clifford group supplemented by up to two canonical $T$ gates.  Compared with the counterpart based on the empirical mean presented in \fref{fig:GHZ} in the main text, the MSE is larger by about 30\%, but the overall behavior is almost the same. Similar results apply to other numerical simulations. So our main conclusions do not change if the empirical mean is replaced by the median of means.

\section{Details of numerical simulation}\label{sup:Simulation}
Classical simulation of a quantum circuit has long been an important task for verification and validation of quantum devices.  The Gottesman-Knill theorem states that a Clifford circuit (also known as a stabilizer circuit) can be efficiently simulated on a classical computer. Later on, Aaronson and Gottesman proposed a practical simulation algorithm \cite{AaroG04}. Using the tableau representation, their algorithm runs in time $\caO(n^2)$ for both deterministic and random measurements, and in time $\caO(n^3)$ for computing the inner product between two stabilizer states, where $n$ is the qubit number \cite{AaroG04}. Later, several alternative approaches were introduced, including simulation methods based on the graph state representation \cite{AndeB06} and canonical forms of Clifford unitaries~\cite{Nest08}. In addition, a number of  sampling algorithms for the Clifford group were proposed~\cite{KoenS14,Berg21,BravM21}. The subgroup algorithm  in \rcite{KoenS14}
was later generalized to the qudit setting \cite{Hein21}. However, as far as we know, no explicit algorithm for simulating a qudit Clifford circuit has been implemented so far.

Towards the simulation of universal quantum circuits,  a number of works have studied  the simulation of a (qubit) Clifford circuit interleaved by magic gates. Currently, there are two main approaches  for dealing with magic gates: the first one is based on gadgetization \cite{PashRKB22,BravG16}, and the second one is based on low-rank stabilizer decomposition \cite{BravBCC19}. For both approaches, the computational cost is polynomial in the qubit number and exponential in  the number of magic gates.

For our purpose, we develop a complete simulation algorithm for qudit shadow estimation based on the Clifford group supplemented by $T$ gates. To this end, we generalize the tableau representation of stabilizer states \cite{AaroG04} and gadgetization method  \cite{BravG16,PashRKB22} to the qudit setting, and combine them with the Clifford sampling algorithm proposed in \rcite{Hein21}.
The simulation algorithm  we develop here may be useful to studying various other problems in quantum information processing and is thus of interest beyond the focus of this work.

\subsection{Generating matrices of qudit stabilizer states}
Before formulating the simulation algorithm,
we need a convenient representation of qudit stabilizer states and stabilizer projectors,  assuming that the local dimension $d$ is an odd prime.  Here we  first introduce a representation based on generating matrices, which is the basis for more sophisticated representation methods. 

Suppose $|\Psi\> \in \Stab(n,d)$ is an $n$-qudit stabilizer state, and $S_1, S_2,...,S_n$ form a minimal generating set of its stabilizer group. The 
$i$th  stabilizer generator can be expressed as follows,
\begin{equation}\label{eq:stabvec}
	S_i = \chi(r_i^u)W_{\bfu_i}= \chi(r^u_i)\bigotimes_{k=1}^n W(u^z_{i,k},u^x_{i,k}),
\end{equation}
where  the (column) vector $\bfu_i:= (u^z_{i,1},\dots,u^z_{i,n},u^x_{i,1},\dots,u^x_{i,n})^\top$ is the  \emph{stabilizer vector} of $|\Psi\>$ associated with $S_i$ and $r^u_i$ is the corresponding \emph{phase exponent}.
The \emph{generating matrix} $\caG$ of $|\Psi\>$ (with respect to the given generating set) is composed of such stabilizer vectors and phase exponents
as follows,
\begin{equation}\label{eq:GenMatrix}
	\caG = \left[\begin{array}{c|c|c}
		\caG^z & \caG^x & r 
	\end{array}\right]
	=\left[\begin{array}{ccc|ccc|c}
		u^z_{1,1} &  \cdots & u^z_{1,n}  &   u^x_{1,1} & \cdots & u^x_{1,n}  &  r^u_1  \\
		\vdots  &  \ddots & \vdots   &  \vdots  &  \ddots & \vdots   &  \vdots \\
		u^z_{n,1} &  \cdots & u^z_{n,n}  &  u^x_{n,1} & \cdots & u^x_{n,n}  &  r^u_{n}
	\end{array}
	\right],  
\end{equation}
where $\caG^z$ and $\caG^x$ encode the exponents of $Z$ and $X$, respectively, while $r=r(\caG)$ encodes the phase exponents. The submatrix of $\caG$ excluding the last column is called the \emph{reduced generating matrix} and denoted by $\tcaG$ henceforth. 

For the convenience of the following discussion, we use $\len(\caG)$ to denote the number of rows in $\caG$ and use $\<\caG\>$ as a shorthand for the stabilizer group $\<S_1,S_2,\ldots,S_n\>$. Note that $|\Psi\>$ and its stabilizer group  are uniquely determined by the generating matrix $\caG$. 
According to the properties of stabilizer generators, if we swap two rows of $\caG$, add one row to another, or multiply one row by any number in $\bbF_d^\times=\bbF_d \setminus\{0\}$, we get another generating matrix (with respect to a different generating set) of the same stabilizer state $|\Psi\>$. We shall call these `row operations'. 

Recall that any Clifford unitary $C\in \Cl(n,d)$ has the form $C=W_\bfg \mu(M)$ with $\bfg\in \bbF_d^{2n}$ and $M\in \Sp(2n,d)$ up to an overall phase factor [see \eref{eq:CaM}]. 
According to \eref{eq:CWC},  applying the Clifford unitary $C$ on $|\Psi\>$ is equivalent to  updating the generating matrix $\caG$ via the mapping 
\begin{equation}\label{eq:update_ur}
	\bfu_i \mapsto M \bfu_i, \quad r^u_i \mapsto r^u_i+[\bfg,M\bfu_i], \quad  i = 1,2,\dots,n.
\end{equation}
In numerical simulation of shadow estimation based on the Clifford group, we need to 
sample a random Clifford unitary. In the generating-matrix representation, this task  amounts  to  sampling a random symplectic matrix $M \in \Sp(2n,d)$ and a random symplectic vector $\bfg \in \bbF_d^n$. Here 
we  use the algorithm proposed in \rcite{Hein21} for sampling a symplectic matrix, which is a generalization of the subgroup algorithm for $\Sp(2n,2)$ proposed in \rcite{KoenS14}. The sampling of a symplectic vector is straightforward. 

The above results can be generalized to stabilizer projectors in a straightforward way. Notably, a stabilizer projector of rank $d^j$ can be represented by a generating matrix with $n-j$ rows. By convention all rows of the generating matrix are linearly independent, and similarly for the reduced generating matrix.

\subsection{Tableau representation of qudit stabilizer states}\label{sec:tableau}
The generating matrix offers a very simple representation of stabilizer states, but it is not so convenient for computing the inner product and for simulating a quantum measurement. To circumvent this problem, here we turn to a more sophisticated representation,  which contains valuable additional information and is more convenient in numerical  simulation. 
The \emph{tableau} $\caJ$ of a stabilizer state $|\Psi\>\in \Stab(n,d)$ is a $2n \times (2n+1)$ matrix over $\bbF_d$. The upper half of $\caJ$ is exactly the generating matrix $\caG$, which is composed of stabilizer vectors $\bfu_i$ and phase exponents $r^u_i$ as shown in \eref{eq:GenMatrix}. The lower half  is composed of $n$ row vectors of the form  $[\bfv_j^\top \,|\, r^v_j]$, where the vectors $\bfv_j\in \bbF_d^{2n}$  satisfy the following relations on symplectic products  \cite{AaroG04},
\begin{equation}\label{eq:commuterelation}
	[\bfu_i,\bfu_j] = 0, \quad [\bfv_i,\bfv_j] = 0, \quad [\bfu_i,\bfv_j] = \delta_{ij}, \quad i,j =1,2,\dots, n,
\end{equation}
and are called  \emph{destabilizer vectors} henceforth. The destabilizer vectors $\bfv_j$ and  corresponding phase exponents $r^v_j$ satisfy the same update rule  as $\bfu_i$ and $r^u_i$ as presented in \eref{eq:update_ur}. In addition, they determine $n$  Weyl operators as follows,
\begin{equation}
	D_j := \chi\bigl(r^v_j\bigr)W_{\bfv_j} = \chi\bigl(r^v_j\bigr)\bigotimes_{k=1}^n W\bigl(v^z_{j,k},v^x_{j,k}\bigr), \quad \quad j =1,2,\dots, n,
\end{equation}
which are called \emph{destabilizer generators}, in contrast with the $n$ stabilizer generators $S_j$ for $j=1,2,\ldots, n$. Apparently, row operations on the upper (or lower) half of the tableau $\caJ$ do not change the underlying stabilizer state $|\Psi\>$. 

For the convenience of the following discussion, the submatrix of $\caJ$ composed of the first $2n$ columns  is called the \emph{reduced tableau} of $|\Psi\>$ and is denoted by  $\tcaJ$.

\subsection{Representation of qudit stabilizer states using Lagrangian subspaces and characteristic vectors}
Stabilizer states can also be represented using Lagrangian subspaces \cite{Gros06,Hein21}, which are very useful for computing the inner products between stabilizer states. Recall that a subspace $\caV$ of $\bbF_d^{2n}$ is a \emph{Lagrangian subspace} if $\caV$ has dimension $n$ and the symplectic product vanishes in this subspace. Let $\caV$ be a subspace of $\bbF_d^{2n}$, the symplectic complement of $\caV$  is defined as 
\begin{align}
	\caV^\perp:=\bigl\{\bfu\in \bbF_d^{2n}\;|\; [\bfu,\bfv]=0\;\forall \bfv\in \caV\bigr\}. 
\end{align}
Note that $\caV^\perp=\caV$ iff $\caV$ is a Lagrangian subspace. 

Now, suppose  $\bfu_1,\bfu_2,\ldots,\bfu_n$ are $n$ linearly independent stabilizer vectors of a stabilizer state $|\Psi\>\in \Stab(n,d)$; then $\caL=\spa(\bfu_1,\bfu_2,\ldots,\bfu_n)$ is a Lagrangian subspace in $\bbF_d^{2n}$. In addition, there exists a vector $\bfm \in \bbF_d^{2n}$,  called a \emph{characteristic vector},  such that the projector $|\Psi\>\<\Psi|$ can be expressed as follows \cite{Gros06},
\begin{equation}\label{eq:Psi}
	|\Psi\>\<\Psi|=|\caL,\bfm\>\<\caL,\bfm| := |\caL|^{-1} \sum_{\bfu \in \caL} \chi([\bfm,\bfu])W_\bfu,
\end{equation}
where $|\caL|=d^n$. 	Note that the choice of $\bfm$ is not unique:
given  another vector $\bfm'$ in $\bbF_d^{2n}$, then $|\caL,\bfm'\>$ and $|\caL,\bfm\>$ represent the same stabilizer state iff
$\bfm'-\bfm \in \caL$. Therefore, 
the characteristic vector $\bfm$ has a $d^n$-fold degeneracy.
Based on the above representation, we can  simulate the Clifford transformation $C=W_\bfg \mu(M)$ with $\bfg\in \bbF_d^{2n}$ and $M\in \Sp(2n,d)$	by the following map
\begin{equation}
	\bfu_i \mapsto M\bfu_i, \quad \bfm \mapsto M\bfm+\bfg.
\end{equation}
In addition,  the overlap between two stabilizer states $|\caL_1,\bfm_1\>$ and $|\caL_2,\bfm_2\>$ reads \cite{Hein21}
\begin{equation}\label{eq:ssoverlap}
	|\<\caL_1,\bfm_1|\caL_2,\bfm_2\>|^2 = 
	\begin{cases}
		d^{\,\dim(\caL_1 \cap \caL_2)-n} & \quad \text{if  }  \bfm_1-\bfm_2 \in (\caL_1\cap \caL_2)^\perp, \\
		0 & \quad \text{otherwise}.
	\end{cases}
\end{equation}
Based on this result, we can also simulate the measurement process.

To apply the above representation in numerical simulation, we need to determine a characteristic vector $\bfm$ of the stabilizer state $|\Psi\>$ given the generating matrix in \eref{eq:GenMatrix}. 
Observe that $\bfm$ should satisfy the underdetermined system of linear equations
\begin{equation}
	[\bfm,\bfu_i] = r^u_i\quad \forall \ i = 1, 2, \dots, n,
\end{equation}
which can be solved by Gaussian elimination. This task can be greatly simplified if the destabilizer vectors $\bfv_i$ of $|\Psi\>$ associated with the stabilizer vectors $\bfu_i$ are available, in which case  $\bfm=-\sum_i r_i^u \bfv_i$ is an ideal choice for the characteristic vector.

\subsection{Simulation of a Clifford circuit}\label{sup:SimStab}
Here we introduce a simple  algorithm for simulating  a random qudit Clifford circuit based on the tableau representation, which generalizes the original simulation algorithm tailored to qubits  \cite{AaroG04}. To this end, first  we develop an algorithm for constructing the tableau representation for an arbitrary qudit stabilizer state. Inspired by the representation of stabilizer states based on  Lagrangian subspaces \cite{Gros06,Hein21}, then we develop a method based on Gaussian elimination for computing  the inner product and simulating the computational-basis measurement, which runs in time $\caO(n^3)$.
The overall simulation algorithm in data acquisition phase is presented in Algorithm~\ref{alg:sim_Clifford} below. The sampling of a random symplectic matrix in $\Sp(2n,d)$ is described in detail in Sec.~5.3.1 of \rcite{Hein21}, which we shall not repeat here. Other relevant notation and subroutines are introduced in the following subsections.

\begin{algorithm}
	\SetAlgoLined
	\SetKwInOut{Input}{Input}
	\SetKwInOut{Output}{Output}
	\Input{Stabilizer vectors $\bfs_1,\bfs_2,\dots,\bfs_n$ and phase exponents $p_1, p_2, \dots,p_n$ of the input state $|\Psi\>$}
	\Output{$n$-dit measurement outcome $\bfx$} 
	Construct a tableau $\caJ$ \tcp*{See \sref{sec:construct tableau}}
	Sample a random $M \in \Sp(2n,d)$ and $\bfg \in \bbF_d^n$\;
	Update $\caJ$ with $M$ and $\bfg$ \tcp*{See \sref{sec:tableau}}
	Sample an outcome $\bfx$ in the computational-basis measurement.  \tcp*[f]{See \sref{sec:zbasismeasure}}
	\caption{Data acquisition for a Clifford circuit}
	\label{alg:sim_Clifford}
\end{algorithm}

\subsubsection{Tableau construction}\label{sec:construct tableau}
Suppose $|\Psi\>\in \Stab(n,d)$ is a stabilizer state whose stabilizer group is generated by the $n$ stabilizer operators $\chi(p_1)W_{\bfs_1},\chi(p_2)W_{\bfs_2},\ldots,\chi(p_n)W_{\bfs_n}$, where $\bfs_1, \bfs_2, \ldots,\bfs_n \in \bbF_d^{2n}$ and $p_1, p_2, \ldots, p_n \in \bbF_d$ are the corresponding stabilizer vectors and phase exponents. Let $\caL=\spa(\bfs_1, \bfs_2, \ldots,\bfs_n)$ be the Lagrangian subspace associated with the stabilizer state $|\Psi\>$. 
Based on this information
we can construct a tableau representation $\caJ$ of $|\Psi\>$ in two steps: (1)~constructing stabilizer and destabilizer vectors; (2) tracking the phase exponents.

\textbf{(1)~Constructing stabilizer and destabilizer vectors:} 
In this step, we want to find  $n$ stabilizer vectors $\bfu_1,\bfu_2,\dots,\bfu_n$ and $n$ destabilizer vectors $\bfv_1,\bfv_2,\dots,\bfv_n$ that satisfy \eref{eq:commuterelation} and
such that  $\{\bfu_i\}_{i=1}^n$ spans the same Lagrangian subspace as $\{\bfs_i\}_{i=1}^n$
This is achieved by Algorithm~\ref{alg:tableau} introduced below. Let $\caL_0$ be the Lagrangian subspace associated with the basis state $|0\>^{\otimes n}$. 
The basic idea of Algorithm~\ref{alg:tableau} can be summarized as follows: find a symplectic transformation that maps $\caL_0$ to $\caL$, then the desired stabilizer vectors $\bfu_1,\bfu_2,\dots,\bfu_n$ and destabilizer vectors  $\bfv_1,\bfv_2,\dots,\bfv_n$ of $|\Psi\>$
can be constructed by applying the same symplectic transformation to the standard  stabilizer vectors  and destabilizer vectors of $|0\>^{\otimes n}$.  The overall computational complexity of Algorithm~\ref{alg:tableau} is $\caO(n^3)$.

The desired symplectic transformation can be expressed as a product of some elementary  symplectic transformations known as \emph{symplectic transvections} \cite{Hein21}. Given a vector $\mathbf{h} \in \bbF_d^{2m}$ (in our case $m=1,2,...,n$) and a scalar $\lambda \in \bbF_d^\times$, a symplectic transvection $T_{\lambda,\mathbf{h}}$ on $\bbF_d^{2m}$ can be defined via the map
\begin{equation}
	T_{\lambda,\mathbf{h}}(\bfx) = \bfx+\lambda[\bfx,\bfh]\bfh.
\end{equation}
The inverse transvection reads $T_{\lambda,\bfh}^{-1}= T_{-\lambda,\bfh}$. 
It is well known that symplectic transvections can generate the whole symplectic group.
More concretely, given any two nonzero vectors $\bfx,\bfy$ in  $\bbF_d^{2m}$,   there exist two vectors $\bfh_1,\bfh_2\in \bbF_d^{2m}$ and scalars $\lambda_1,\lambda_2\in\bbF_d$ such that 
\begin{equation}
	\bfy = T_{\lambda_1,\bfh_1}T_{\lambda_2,\bfh_2}(\bfx).
\end{equation}
This is Lemma 5.1 in \rcite{Hein21},  which also offers a constructive proof. So  it should be clear how the subroutine `find transvections' in Algorithm~\ref{alg:tableau} works. To simplify the  description,
the product $T_{\lambda_1,\bfh_1}T_{\lambda_2,\bfh_2}$ is abbreviated as $T$ in the pseudocode. Here we take the `local' convention, which means the
$(2j-1)$th and $2j$th entries of $\bfu$ denote the powers of $Z_j$ and $X_j$, respectively.
Let $\bfe_{2m}(i)\in \bbF_d^{2m}$ be the column vector with the $i$th entry equal to 1 and all other entries equal to 0. Note that the symplectic products in \eref{eq:commuterelation} are preserved under symplectic transformations. In each iteration in Algorithm~\ref{alg:tableau}, we have
$\bigl[\bfu^{(i+1)}_j, \bfe_{2(n\!+\!1\!-\!i)}(1)\bigr]=0$ for $j=i+1,i+2,\ldots, n$, which means the second entry of $\bfu^{(i+1)}_j$ must be zero. In addition, we can  set the first entry to zero without changing the Lagrangian subspace of interest. Furthermore, we can discard the first two entries and continue to look for transvections in  the symplectic group $\Sp(2(n\!-\!i),d)$ in the next iteration.

\begin{algorithm}
	\SetAlgoLined
	\SetKwInOut{Input}{Input}
	\SetKwInOut{Output}{Output}
	\Input{Stabilizer vectors $\bfs_1,\bfs_2,\dots,\bfs_n$}
	\Output{Stabilizer vectors  $\bfu_1,\bfu_2,\dots,\bfu_n$ and destabilizer vectors $\bfv_1,\bfv_2,\dots,\bfv_n$}
	Set $\bigl\{\bfu^{(1)}_1,\dots,\bfu^{(1)}_n\bigr\} =\{ \bfs_1,\bfs_2,\dots,\bfs_n\}$\;
	\For{$i = 1:n$}{
		Find transvections and construct $T_i$ such that $T_i\bfe_{2(n\!+\!1\!-\!i)}(1)=\bfu^{(i)}_i$\;
		\For{$j = i+1:n$}{
			$\bfu^{(i+1)}_j = T_i^{-1}\bfu^{(i)}_j$ and discard the first two entries of $\bfu^{(i+1)}_j$\;
		}
	} 
	$M=T_n$\;
	\For{$i = (n-1):1$}{
		$M =T_i (\mathds{1}_2 \bigoplus M)$\;
	}
	$\bfu_1,\bfu_2,\dots,\bfu_n$ and $\bfv_1,\bfv_2,\dots,\bfv_n$  correspond to the first $n$ columns and last $n$ columns of $M$, respectively.
	\caption{Constructing the reduced tableau}\label{alg:tableau}
\end{algorithm}

\textbf{(2) Tracking phase exponents:} Now we determine the phase exponents in the last column of the tableau $\caJ$. Note that $\caL=\spa(\bfs_1,\bfs_2,\dots,\bfs_n)=\spa(\bfu_1,\bfu_2,\dots,\bfu_n)$, and every $\bfs_i$  can be expressed as  a linear combination of $\bfu_1,\bfu_2,\dots,\bfu_n$ as follows, $\bfs_i = \sum_{j=1}^n c^j_i \bfu_j$, where
\begin{equation}
	c^j_i = [\bfs_i, \bfv_j]. 
\end{equation}
Therefore, the phase exponents $r^u_1,r^u_2,\dots, r^u_n$ in the last column of $\caJ$ should satisfy the  following system of linear equations:
\begin{equation}
	\sum_{j=1}^n c^j_i r^u_j = p_i, \quad i = 1,2,\dots,n,
\end{equation}
which can be solved in time  $\caO(n^3)$ by  Gaussian elimination. The phase exponents $\{r^v_i\}_{i=1}^n$ associated with destabilizer vectors  are actually redundant. For simplicity, we set as convention $r^v_{i} = 0$ for $i = 1, 2, \dots,n$.

\begin{algorithm}[tbp]
	\SetAlgoLined
	\SetKwInOut{Input}{Input}
	\SetKwInOut{Output}{Output}
	\Input{The Lagrangian subspace $\caL$ and characteristic vector $\bfm$ of the pre-measurement state}
	\Output{Measurement outcome $\bfx$}
	Find a basis $\{\mathbf{w}_i\}_{i=1}^{k_*}$ for  $\caL_0\cap \caL$ and add $n-k_*$ vectors $\bfy_1, \bfy_2,\ldots, \bfy_{n-k_*}$ to construct a basis for $\caL_0$\;
	Construct the equations according to \eref{eq:xui} and solve for $\tbfx$\;
	Return the latter half of $\tbfx$ as $\bfx$.
	\caption{Simulation of the computational-basis measurement}\label{alg:measurement}
\end{algorithm}

\subsubsection{Computational-basis measurement}\label{sec:zbasismeasure}
The simulation of a computational-basis measurement can be seen as a variant of computing the inner product of two stabilizer states. For this task it is more convenient to go back to the `global' convention and to use the representation based on Lagrangian subspaces and characteristic vectors.

Suppose the stabilizer state before the measurement is $|\Psi\>=|\caL,\bfm\>$, where $\caL$ and $\bfm$ are the corresponding Lagrangian subspace and characteristic vector. 
The  outcome  of the computational-basis measurement can be labeled by a vector $\bfx$ in $\bbF_d^n$, which corresponds to the basis state $|\bfx\>$. This basis state can also be expressed as
$|\bfx\>=|\caL_0,\tbfx\>$, where $\caL_0$ is  the Lagrangian subspace
associated with the computational basis and is
spanned by the $n$ vectors $\bfe_{2n}(1),\bfe_{2n}(2),...,\bfe_{2n}(n)$, and the characteristic vector $\tbfx$ can be chosen to be $\tbfx=(\mathbf{0};\bfx)$, where ";" means vertical concatenation instead of horizontal concatenation. If we perform the computational-basis measurement on the stabilizer state $|\caL,\bfm\>$, then the probability of obtaining outcome $\bfx$ reads
\begin{equation}
	p(\bfx)=|\<\caL,\bfm|\caL_0,\tbfx\>|^2.
\end{equation}

Next, we  introduce a scheme for sampling the $n$-dit measurement outcome $\bfx$. First, construct a basis $\{\bfw_i\}_{i=1}^{k_*}$ for $\caL_0\cap \caL$ using  the Zassenhaus algorithm \cite{LuksRW97} for example, where $k_* = \dim (\caL_0\cap \caL)$ and $0 \leq k_* \leq n$. Choose $n-k_*$ vectors $\bfy_1, \bfy_2,\ldots, \bfy_{n-k_*}$ in $\bbF_d^{2n}$ such that $\{\bfw_i\}_{i=1}^{k_*}\cup \{\bfy_j\}_{j=1}^{n-k_*}$ forms a basis for $\caL_0$. Then we can sample the measurement outcome $\bfx$ by solving the system of linear equations 
\begin{equation}\label{eq:xui}
	\begin{aligned}
		[\tbfx, \bfw_i] &= [\bfm,\bfw_i] \quad \text{for }  i=1, 2, \dots, k_*, \\
		[\tbfx, \bfy_j] &= c_{j}  \quad \text{for }  j=1,2,\dots, n-k_*,
	\end{aligned}
\end{equation}
where each $c_j$ is a random number in $\bbF_d$. Finally, we take the latter half of $\tbfx$, which is exactly $\bfx$. The overall procedure is summarized in Algorithm~\ref{alg:measurement}.

\subsection{Simulation of a Clifford circuit supplemented by $T$ gates}
\label{sec:Clifford+T}
In the main text, we focus on fidelity estimation using shadow estimation as a prototypical task. In this section, we shall describe how to classically simulate this procedure when the unitary ensemble corresponds to a random Clifford circuit supplemented by $T$ gates.

In the data acquisition stage, an input state goes through a layer of random Clifford gate $C_0 \in \Cl(n,d)$, a layer of $T$ gates, and a layer of Fourier gates, followed by the computational-basis measurement. If the input state can be prepared from $|0\>^{\otimes n}$ with a few layers of Clifford and $T$ gates and the total number of $T$ gates increases at most logarithmically with $n$, then this procedure can be efficiently simulated. That is to say, we can sample the measurement outcome $\bfx:=(x_1,x_2,\dots,x_n)^\top \in \bbF_d^n$ of this quantum circuit according to the probability distribution $\{p(\bfx)\}_{\bfx}$ on a classical computer, with computational cost polynomial in $n$ and exponential in  $t=t_1+t_2$, where $t_1$ is the number of $T$ gates used to prepare the input state, and $t_2$ is the number of $T$ gates used to realize the unitary ensemble for shadow estimation.

\subsubsection{Reverse gadgetization}
To simulate the circuit mentioned above, we generalize the gadgetization method  in \rcite{PashRKB22} to the qudit setting. For simplicity, we denote each diagonal magic gate liberally as $T$, regardless of its specific type.
For each $T$ gate we can define a single-qudit magic state $|T^\dag\>$ according to \eref{eq:TdagMagic}. 
Then the $T$ gate can be replaced equivalently by a `reversed' gadget, which is essentially an injection of the magic state $|T^\dag\>$ using the $CX$ gate in \eref{eq:Clgen}, followed by post-selection. 
Diagrammatically, the reversed gadget can be expressed as follows,
\begin{equation}
	\includegraphics[width=0.35\textwidth]{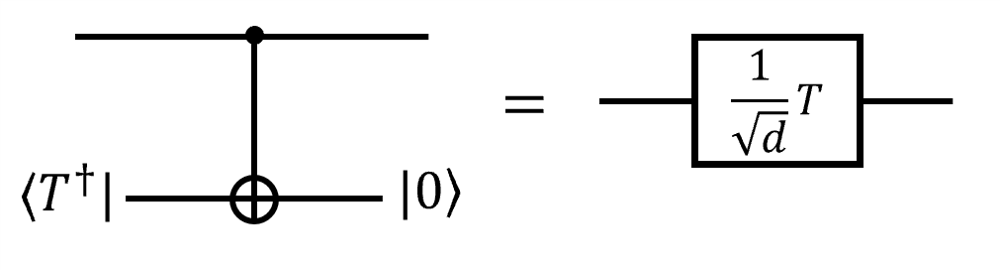}.
\end{equation}
Circuits in this section are read from right to left, so that we can derive the expressions of  intermediate quantum states directly from the circuit diagrams.

Suppose the input state is $ C_1 T^{\otimes s_1} \otimes C_2 T^{\otimes s_2}\cdots C_{l-1} T^{\otimes s_{l-1}}C_{l}|0\>^{\otimes n}$, where  $C_1,C_2,...,C_l \in \Cl(n,d)$ and $\sum_{i=1}^{l-1} s_i = t_1$, then the circuit $U$ we are to simulate can be expressed diagrammatically as in the upper plot in \fref{fig:qc}.  We assume that the $t_2$ qudit $T$ gates in the measurement primitive are applied to the first $t_2$  qudits without loss of generality. 
Since $C_0$ is selected randomly from $\Cl(n,d)$, $C_{1}$ can be absorbed into $C_0$.  The equivalent circuit based on  reverse gadgetization is shown in the lower plot in \fref{fig:qc}.

\begin{figure}[tbp]
	\centering
	\includegraphics[width=0.95\textwidth]{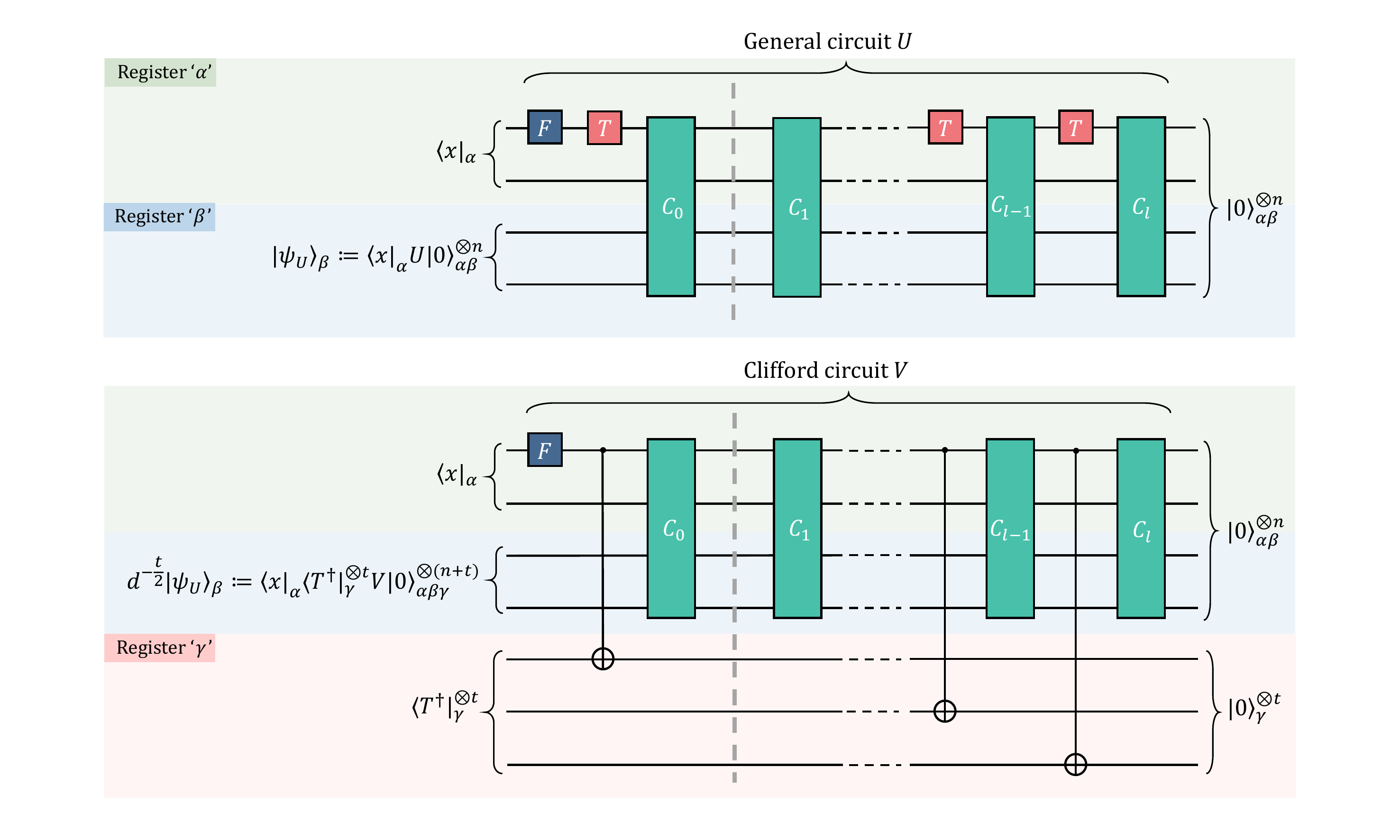}
	\caption{(upper) A Clifford circuit interleaved by $T$ gates. The gates on the right of the vertical dashed line are used to prepare the input state, while the gates on the left are used to realize a randomized quantum measurement employed in shadow estimation. 
		(lower) The post-selected circuit obtained after reverse gadgetization of all $T$ gates.}
	\label{fig:qc}
\end{figure}

\subsubsection{Main simulation algorithm}
To sample the measurement outcome $\bfx$, we sample the outcome on each qudit sequentially \cite{BravG16}, in which the sampling on the $j$th qudit with  $j = 2,3,\dots,n$ is determined by the conditional probabilities
\begin{equation}
	p(y\,|\,x_1,\dots,x_{j-1}):= \frac{p(x_1,\dots,x_{j-1},y)}{p(x_1,\dots,x_{j-1})},\quad y\in \bbF_d.
\end{equation}
In this way, to simulate each experimental shot, we  need to calculate the outcome probabilities $p(\bfx)$ for only $nd$ sequences $\bfx$ of length varying from $1$ to $n$, instead of all $d^n$ probabilities $p(\bfx)$ for $\bfx\in \bbF_d^n$.

Next, we illustrate how to evaluate the  probability $p(\bfx)$ of  the outcome $\bfx = (x_1,\dots,x_m)^\top$ of a measurement on the first  $m$ qudits with $1\leq m\leq n$. For the convenience of the following discussion, the first $m$ qudits are referred to as the measured register `$\alpha$', and the remaining $(n-m)$ qudits  are referred to as the marginalized register `$\beta$'; the identity operators on registers $\alpha$ and $\beta$ are denoted by $\bbI_\alpha$ and $\bbI_\beta$, respectively. Then the probability $p(\bfx)$ can be expressed as follows,
\begin{equation}
	p(\bfx) = \bigl\|\<\bfx|_\alpha U |0\>_{\alpha \beta}^{\otimes n}\bigr\|_2^2.
\end{equation}
After gadgetization, the circuit $U$ acting on $|0\>_{\alpha\beta}^{\otimes n}$ is re-expressed as a Clifford circuit $V$ acting on $|0\>_{\alpha\beta\gamma}^{\otimes (n+t)}$, with $t$ ancillary qudits in register `$\gamma$' post-selected on the state $|T^\dag\>^{\otimes t}$. Thus the above probability can be expressed as 
\begin{equation}\label{eq:px1}
	p(\bfx) = d^t \bigl\|\<\bfx|_\alpha \<T^\dag|_\gamma^{\otimes t} V |0\>_{\alpha\beta\gamma}^{\otimes (n+t)}\bigr\|_2^2=d^t\tr\left(\Pi_{\caG_0} \Pi_{\bfx,T}\right),
\end{equation}
where $\caG_{0}$ is the generating matrix of the stabilizer state $V|0\>_{\alpha\beta\gamma}^{\otimes n+t}$ and
\begin{equation}
	\Pi_{\caG_0} = V|0\>\<0|_{\alpha\beta\gamma}^{\otimes (n+t)}V^\dag, \quad \Pi_{\bfx,T} = |\bfx\>\<\bfx|_\alpha \otimes \bbI_\beta \otimes \bigl(|T^\dag\>\<T^\dag|\bigr)_\gamma^{\otimes t}.
\end{equation}

Let 
\begin{align}\label{eq:Pixgamma}
	\Lambda_{\bfx,\gamma}=\tr_{\alpha\beta}\bigl[\Pi_{\caG_0} \bigl(|\bfx\>\<\bfx|_\alpha \otimes \bbI_{\beta\gamma}\bigr)\bigr],
\end{align}
where $\bbI_{\beta\gamma}$ is the identity operator on registers $\beta$ and $\gamma$ together. Then the probability $p(\bfx)$ can be expressed as $p(\bfx)=d^t\tr[\Lambda_{\bfx,\gamma}(|T^\dag\>\<T^\dag|)_\gamma^{\otimes t}]$. If the partial trace in \eref{eq:Pixgamma} does not vanish, then $\Lambda_{\bfx,\gamma}$ is proportional to 
a stabilizer projector $\Pi_{\caG_{\bfx,\gamma}}$  acting on register $\gamma$, where $\caG_{\bfx,\gamma}$ is a generating matrix.  Eventually, $p(\bfx)$ can be expressed as follows,
\begin{equation}\label{eq:px2}
	p(\bfx)= 
	\begin{cases}
		d^{\xi-m}\tr\left[\Pi_{\caG_{\bfx,\gamma}} \bigl(|T^\dag\>\<T^\dag|\bigr)_\gamma^{\otimes t} \right] & \quad \text{if  } flag=\text{True}, \\
		0  & \quad \text{if  } flag=\text{False}.
	\end{cases}
\end{equation}
Here the generating matrix $\caG_{\bfx,\gamma}$ and relevant parameters, including $\xi$ and $flag$, can be determined by virtue of  Algorithm~\ref{alg:ConStab}  in \sref{sup:PartialTrace} below. 
After obtaining $\caG_{\bfx,\gamma}$, the probability $p(\bfx)$ can be computed directly in time $\caO\bigl(td^{k_\gamma+1}\bigr)$ with $k_\gamma = \len(\caG_{\bfx,\gamma}) \leq t$, where $t$ is the total number of single-qudit magic gates employed in the simulation; the computational cost usually increases exponentially with $t$.

The procedure for simulating a Clifford circuit supplemented by $T$ gates is summarized in Algorithm~\ref{alg:sim_Clifford+T}, where the `Constraining the stabilizers' subroutine in Line 8 is explained in detail in \sref{sup:PartialTrace} below. Naive implementation of Algorithm~\ref{alg:sim_Clifford+T} results in a time complexity of $\caO\bigl(nd(n+t)^3+ntd^{t+2}\bigr)$. However, we can make some optimization by reusing some intermediate results, which can reduce the complexity to $\caO\bigl((n+t)^3+(nd+t)d^{t+1}\bigr)$.

\begin{algorithm}[tb]
	\SetAlgoLined
	\SetKwInOut{Input}{Input}
	\SetKwInOut{Output}{Output}
	\Input{Circuit representation of $U$}
	\Output{$n$-dit measurement outcome $\bfx$} 
	Construct the generating matrix $\caG_0$ for $|0\>^{\otimes (n+t)}$ and apply $C_{l}, CX, C_{l-1},...$ sequentially\;
	Sample a random $M \in \Sp(2n,d)$ and $\bfg \in \bbF_d^{2n}$\;
	Update $\caG_0$ with $M$ and $\bfg$\;
	
	Update $\caG_0$ according to the actions of  $CX$ and $F$ gates on the relevant qudits\;
	\For{$m=1:n$}{
		\For{$y=0:d-1$}{
			$\bfy$ = concatenate($\bfx,y$)\;
			$flag$, $\xi$, $\caG_{\bfy,\gamma}$ = Constraining the stabilizers($\caG_0$,$\bfy$)\;
			\eIf{$flag$= True,}
			{$p(\bfy) =d^{\xi-m}\tr\bigl[\Pi_{\caG_{\bfy,\gamma}} (|T^\dag\>\<T^\dag|)_\gamma^{\otimes t} \bigr]$\;}
			{$p(\bfy) = 0$\;}
		}
		Sample $\bfy$ according to $\{p(\bfy)/p(\bfx)\}_{y=0}^{d-1}$, $\bfx \leftarrow \bfy$\;
	}
	\Return{$\bfx$.}
	\caption{Data acquisition for a Clifford circuit supplemented by $T$ gates}
	\label{alg:sim_Clifford+T}
\end{algorithm}

In principle, the Clifford circuit discussed in the previous section can also be simulated using Algorithm~\ref{alg:sim_Clifford+T} with some minor modification. When the number of $T$ gates to be gadgetized is $t=0$, the outcome probability $p(\bfx)$ is given directly by Algorithm~\ref{alg:ConStab} as
\begin{equation}
	p(\bfx)= 
	\begin{cases}
		d^{\xi-m} & \quad \text{if  } flag=\text{True}, \\
		0  & \quad \text{if  } flag=\text{False}.
	\end{cases}
\end{equation}
The corresponding time complexity for sampling one-dit outcome is $\caO(n^3d)$, and the overall complexity for sampling an $n$-dit outcome $\bfx$ is approximately $\caO(n^4d)$ by naive implementation. By reusing some intermediate results, we can reduce the overall complexity to $\caO(n^3d)$, but it is still not as efficient as the tableau method in \sref{sup:SimStab}. In addition, when no $T$ gates are involved, it is much more efficient to calculate the inner product between two stabilizer states using Lagrangian subspaces and characteristic vectors.

\subsubsection{Partial trace of stabilizer projectors}\label{sup:PartialTrace}
In this subsection we describe how to compute the partial trace in \eref{eq:Pixgamma}, which corresponds to the step `Constraining the stabilizers' in Algorithm~\ref{alg:sim_Clifford+T}. To this end, we need to  introduce some additional notation.
For any Weyl operator $W=\omega_d^a W_{u_1}\otimes \cdots \otimes W_{u_n}   \in \caW(n,d)$ with  $a \in \bbF_d$ and $u_1,u_2,\ldots,u_n \in \bbF_d^2$, we use $|W|_i$ to denote the $i$th tensor factor $W_{u_i}$ and $\omega(W)$ to denote the phase factor $\omega_d^a$. In addition, we use $|W|_\alpha$ to denote the sub-string of tensor factors associated with register `$\alpha$', that is, $|W|_\alpha:=\bigotimes_{i\in \alpha} W_{u_i}=\bigotimes_{i=1}^m W_{u_i}$, and similarly for $|W|_\beta$ and $|W|_\gamma$.

In order to contribute non-trivially to the partial trace in \eref{eq:Pixgamma}, a stabilizer operator $W\in \<\caG_{0}\>$ must satisfy the following constraints: 
\begin{enumerate}
	\item $|W|_i=I$ for each qudit  in register `$\beta$', that is, $m+1 \leq i \leq n$.	
	
	\item $|W|_i \in \{I,Z,\ldots,Z^{d-1}\}$ for each qudit in register `$\alpha$', that is, $1 \leq i \leq m$. 
\end{enumerate}
Note that the subset of stabilizer operators in $\<\caG_0\>$ that satisfy the above two constraints also forms a stabilizer group. Let $\caG_2$ be a generating matrix associated with the resulting stabilizer group and $\Pi_{\caG_2}$ the corresponding stabilizer projector. Let  $\xi=\len(\caG_2)$ be the number of rows of $\caG_2$ and let $S_j$ be the stabilizer operator associated with the $j$th row of $\caG_2$. Then $\Lambda_{\bfx,\gamma}$ in \eref{eq:Pixgamma} can be expressed as follows,
\begin{align}\label{eq:Pixgamma2}
	\Lambda_{\bfx,\gamma}&=\tr_{\alpha\beta}\bigl[\Pi_{\caG_0} \bigl(|\bfx\>\<\bfx|_\alpha \otimes \bbI_{\beta\gamma}\bigr)\bigr]
	=d^{\xi-n-t}\tr_{\alpha\beta}\bigl[\Pi_{\caG_2} \bigl(|\bfx\>\<\bfx|_\alpha \otimes \bbI_{\beta\gamma}\bigr)\bigr]\nonumber\\
	&=
	d^{-m-t}\sum_{W\in \<\caG_2\>}h_\bfx(W)=
	d^{-m-t}\prod_{j=1}^\xi \sum_{k_j\in \bbF_d}h_\bfx\Bigl(S_j^{k_j}\Bigr),
\end{align}
where $h_\bfx$ is a function from the stabilizer group $\<\caG_2\>$ to the Pauli group on register $\gamma$,
\begin{equation}
	h_\bfx(W):=d^{m-n}\tr_{\alpha\beta}\bigl[W \bigl(|\bfx\>\<\bfx|_\alpha \otimes \bbI_{\beta\gamma}\bigr)\bigr]=\omega(W) \<\bfx||W|_\alpha |\bfx\>  |W|_\gamma.
\end{equation}
Note that $\Lambda_{\bfx,\gamma}=0$ whenever there exists a stabilizer operator $W\in \<\caG_2\>$ such that $h_\bfx(W)=\omega^k\bbI_\gamma$ for some $k\in \bbF_d^\times$, where $\bbI_\gamma$ is the identity operator on register $\gamma$. 

Now, by performing suitable row operations on $\caG_2$ we can construct a new generating matrix $\caG_3$ such that $\tilde{\caG}_3[:,n+1\!:\!n+t]$ has right row echelon form. Then $\Lambda_{\bfx,\gamma}=0$ iff $h_\bfx(S_j')=\omega^k\bbI_\gamma$ for some stabilizer generator $S_j'$ associated with $\caG_3$ and $k\in \bbF_d^\times$. If such a stabilizer generator does not exist,  then $\Lambda_{\bfx,\gamma}$ is proportional to the stabilizer projector associated with the following set of stabilizer generators acting on register $\gamma$,
\begin{align}
	\bigl\{h_\bfx\bigl(S_j'\bigr)\; \big|\; h_\bfx\bigl(S_j'\bigr)\neq \bbI_\gamma\bigr\},
\end{align}
which can be encoded by a generating matrix $\caG_{\bfx,\gamma}$. More precisely, we have  
\begin{align}
	\Lambda_{\bfx,\gamma}=d^{\xi-m-t}\Pi_{\caG_{\bfx,\gamma}}.
\end{align}
Based on this observation we can devise an efficient algorithm for computing the partial trace in \eref{eq:Pixgamma2} and for determining the generating matrix $\caG_{\bfx,\gamma}$. Before presenting the algorithm, it is convenient to introduce three useful functions.

\textbf{(1) Slicing of $\caG$}: We denote by $\caG^z[i\!:\!j,k\!:\!l]$ the sub-matrix of $\caG^z$ that contains rows from $i$ to $j$ and columns from $k$ to $l$, and similarly for $\caG^x[i\!:\!j,k\!:\!l]$. For $\caG$ we define
\begin{equation}\label{eq:sliceG}
	\caG[i\!:\!j,k\!:\!l] := 
	\left[\begin{array}{ccc|ccc|c}
		u^z_{i,k} &  \cdots & u^z_{i,l}  &  u^x_{i,k} & \cdots & u^x_{i,l} & r^u_{i} \\
		\vdots  &  \ddots & \vdots  & \vdots  &  \ddots & \vdots  & \vdots \\
		u^z_{j,k} &  \cdots & u^z_{j,l} & u^x_{j,k} & \cdots & u^x_{j,l} & r^u_{j}
	\end{array}
	\right].
\end{equation}
When the argument in position $i \, (k)$ is omitted, we mean starting from the first row (column). When the argument in position $j \, (l)$ is omitted, we mean stopping on the last row (column).
Moreover, $\tcaG[i\!:\!j,k\!:\!l]$ denotes the submatrix of $\caG[i\!:\!j,k\!:\!l]$ without the last column, while $r(\caG)[i\!:\!j]$ denotes the last column of $\caG[i\!:\!j,k\!:\!l]$, that is,
\begin{equation}
	\tcaG[i\!:\!j,k\!:\!l] := 
	\left[\begin{array}{ccc|ccc}
		u^z_{i,k} &  \cdots & u^z_{i,l}  &  u^x_{i,k} & \cdots & u^x_{i,l}  \\
		\vdots  &  \ddots & \vdots & \vdots  &  \ddots & \vdots \\
		u^z_{j,k} &  \cdots & u^z_{j,l} & u^x_{j,k} & \cdots & u^x_{j,l}
	\end{array}
	\right],\quad 	
	r(\caG)[i\!:\!j]:=
	\begin{bmatrix}
		r^u_{i} \\ \vdots \\ r^u_{j}
	\end{bmatrix}.
\end{equation}

Suppose $i\leq j$ and $k\leq l$; we use $\caG^z[i\!:\!j,l\!:\!k]$ ($\caG^x[i\!:\!j,l\!:\!k]$) to denote the matrix obtained by rearranging the columns of $\caG^z[i\!:\!j,k\!:\!l]$ ($\caG^x[i\!:\!j,k\!:\!l]$) in reverse order, and define 
\begin{equation}
	\begin{aligned}
		\tcaG[i\!:\!j,l\!:\!k]&:=
		\left[\begin{array}{c|c}
			\caG^z[i\!:\!j,l\!:\!k] & \caG^x[i\!:\!j,l\!:\!k]
		\end{array}\right], \\
		\caG[i\!:\!j,l\!:\!k]&:=
		\left[\begin{array}{c|c|c}
			\caG^z[i\!:\!j,l\!:\!k] & \caG^x[i\!:\!j,l\!:\!k] & r(\caG)[i\!:\!j] 
		\end{array}\right].
	\end{aligned}
\end{equation}
Essentially, this means relabeling the qudits in reverse order.

\textbf{(2) Echelon forms of $\caG$}: We denote by $\Echz(\caG,[i\!:\!j,k\!:\!l])$ ($\Echx(\caG,[i\!:\!j,k\!:\!l])$) the generating matrix $\caG'$ obtained by row operations (swap, addition, and multiplication by integers in $\bbF_d^\times$) on $\caG$ such that $\caG'^z[i\!:\!j,k\!:\!l]$  ($\caG'^x[i\!:\!j,k\!:\!l]$) has (left row) echelon form. Here we assume that
only rows  from $i$ to $j$ are involved in the row operations, although 
all columns are updated simultaneously. Again, indices can be omitted if starting from the first row (column) or ending on the last row (column). In addition, we denote by $\REchz(\caG,[i\!:\!j,k\!:\!l])$ ($\REchx(\caG,[i\!:\!j,k\!:\!l])$) the generating matrix $\caG'$ obtained by row operations on $\caG$ such that $\caG'^z[i\!:\!j,k\!:\!l]$ ($\caG'^x[i\!:\!j,k\!:\!l]$) has right row echelon form, that is,  $\caG'^z[i\!:\!j,l\!:\!k]$ ($\caG'^x[i\!:\!j,l\!:\!k]$) has (left row) echelon form.

\textbf{(3) Cut function}: The cut function $\cutz(\caG,[i\!:\!j,k\!:\!l])$ is defined as the smallest row index  $c$  of $\caG$ with $i\leq c \leq j$ such that $\caG^z[c,k\!:\!l]$  has all entries equal to zero (if no row in this range has this property, then we set $\cutz(\caG,[i\!:\!j,k\!:\!l])=j+1$). In a similar way, we can define $\cutx(\caG,[i\!:\!j,k\!:\!l])$ by replacing $\caG^z$ with $\caG^x$.

As an illustration of the above functions,  we can update $\caG$ by row operations such that $\caG^x[:,k\!:\!l]$ has echelon form and then take the rows of $\caG$ with indices $c \geq \cutx(\caG,[:,k\!:\!l])$ to form a new generating matrix $\caG'$. In this way we can make sure that $|W|_i \in \{I,Z,\ldots,Z^{d-1}\}$ for  $W \in \<\caG'\>$ and  $k\leq i \leq l$. By applying  a similar procedure twice, we can construct the aforementioned generating matrix $\caG_2$. 

\begin{algorithm}[tb]
	\SetAlgoLined
	\SetKwInOut{Input}{Input}
	\SetKwInOut{Output}{Output}
	\Input{Generating matrix $\caG_0$, $m$-dit outcome $\bfx$}
	\Output{Boolean $flag$, integer $\xi$, generating matrix $\caG_{\bfx,\gamma}$}
	$\caG_0 \leftarrow \Echx(\caG_0,[:,1\!:\!n])$\;
	$c_1 =\cutx(\caG_0,[:,1\!:\!n])$, $\caG_1 = \caG_0[c_1\!:,:]$\;
	$\caG_1\leftarrow \REchz(\caG_1,[:,1\!:\!n])$\;
	$c_2 =\cutz(\caG_1,[:,m+1\!:\!n])$, $\caG_2 = \caG_1[c_2\!:,:]$\;
	$\caG_3\leftarrow\REchx(\caG_2,[:,n+1\!:\!n+t])$\;
	$c_3 =\cutx(\caG_3,[:,n+1\!:\!n+t])$\;
	$\caG_3\leftarrow\REchz(\caG_3,[c_3\!:,n+1\!:\!n+t])$\;
	$c_4 =\cutz(\caG_3,[c_3\!:,n+1\!:\!n+t])$\;
	$r(\caG_3)\leftarrow r(\caG_3)+ \caG^z_3[:,:\!m]\cdot \bfx$\;
	\eIf{$r(\caG_3)[c_4 \!:]$ contains a non-zero entry,}
	{$flag = \mathrm{False}$\;}{$flag = \mathrm{True}$\;}
	$\caG_{\bfx,\gamma} = \caG_3[:\!c_4-1,n+1\!:\! n+t]$\;
	\Return{$flag$, $\xi = \len(\caG_2)$, $\caG_{\bfx,\gamma}$.}
	\caption{Constraining the stabilizers}\label{alg:ConStab}
\end{algorithm}

The procedure for computing the partial trace in \eref{eq:Pixgamma2} and for determining the generating matrix $\caG_{\bfx,\gamma}$ is summarized in Algorithm~\ref{alg:ConStab}. 
Since the most time-consuming part  is the construction of echelon forms of generating matrices, which is based on Gaussian elimination, the time complexity of Algorithm~\ref{alg:ConStab} is approximately $\caO\bigl((n+t)^3\bigr)$. Note that some intermediate results featured in Algorithm~\ref{alg:ConStab}, including the generating matrices $\caG_1$, $\caG_2$, and $\caG_3$, can be used repeatedly in the sampling algorithm, namely, Algorithm~\ref{alg:sim_Clifford+T}. Based on this observation we can optimize our sampling algorithm, so as to reduce the time complexity. 

\end{document}